\documentclass[10pt,journal,compsoc]{IEEEtran}
\ifCLASSOPTIONcompsoc
  \usepackage[nocompress]{cite}
\else
  \usepackage{cite}
\fi
\usepackage{times}
\usepackage{xcolor}
\usepackage{soul}
\usepackage[utf8]{inputenc}
\usepackage[small]{caption}

\usepackage{lipsum}
\usepackage{amsmath}
\usepackage{amssymb}
\usepackage{amsthm}
\usepackage{bm}
\usepackage{graphicx}
\usepackage{algorithmic}
\usepackage{algorithm}

\usepackage{makecell}
\usepackage{array}
\usepackage{amsthm}
\usepackage{amssymb}
\usepackage{amsmath}
\usepackage{graphicx}  
\usepackage{booktabs}
\usepackage{algorithmic}
\usepackage{algorithm}
\usepackage{helvet}  
\usepackage{multirow}
\usepackage{subfig}
\usepackage{courier}  
\usepackage{url}  

\newtheorem{theorem}{Theorem}

\ifCLASSINFOpdf
\else
\fi
\begin{document}
%
\title{Exact Indexing of Time Series under \\ Dynamic Time Warping}
%
%
%
%

\author{Zhengxin~Li
\IEEEcompsocitemizethanks{\IEEEcompsocthanksitem Z. Li was with the College
of Equipment Management and UAV Engineering, Air Force Engineering University, Xi'an 710051, Shaanxi, P. R. China.
He is currently a postdoctor in the School of Computer Science
and Center for OPTical IMagery Analysis and Learning (OPTIMAL),
Northwestern Polytechnical University.\protect\\
E-mail: zhengxinli@nwpu.edu.cn
}
\thanks{Manuscript first submitted on July 9, 2019.}}


%
%

\markboth{Journal of \LaTeX\ Class Files,~Vol.~14, No.~8, August~2015}%
{Shell \MakeLowercase{\textit{et al.}}: Bare Demo of IEEEtran.cls for Computer Society Journals}
%



\IEEEtitleabstractindextext{%
\begin{abstract}

Dynamic time warping (DTW) is a robust similarity measure of time series.
However, it does not satisfy triangular inequality and has high computational complexity,
severely limiting its applications in similarity search on large-scale datasets.
Usually, we resort to lower bounding distances to speed up similarity search under DTW.
Unfortunately, there is still a lack of an effective lower bounding distance
that can measure unequal-length time series and has desirable tightness.
In the paper, we propose a novel lower bounding distance LB\_Keogh$^+$,
which is a seamless combination of sequence extension and LB\_Keogh.
It can be used for unequal-length sequences and has low computational complexity.
Besides, LB\_Keogh$^+$ can extend sequences to an arbitrary suitable length,
without significantly reducing tightness.
Next, based on LB\_Keogh$^+$, an exact index of time series under DTW is devised.
Then, we introduce several theorems and complete the relevant proofs
to guarantee no false dismissals in our similarity search.
Finally, extensive experiments are conducted on real-world datasets.
Experimental results indicate that
our proposed method can perform similarity search of unequal-length sequences
with high tightness and good pruning power.

\end{abstract}

}

\maketitle

\IEEEdisplaynontitleabstractindextext

%
\IEEEpeerreviewmaketitle

\IEEEraisesectionheading{\section{Introduction}\label{sec:introduction}}

\IEEEPARstart{W}{ith} the rapid development of information technology,
time series pervades almost every field of human activity,
such as finance~\cite{Sezer2018Algorithmic}, traffic~\cite{Hong2018A}, medicine~\cite{Pradhan2017Association},
meteorology~\cite{Wang2018A}, hydrology~\cite{Vesakoski2017Arctic}, multimedia~\cite{Zhao2017Real}, \emph{etc}.
Similarity search is one of the most fundamental problems in time series data mining~\cite{Qin2018Salient}.
It can be applied in many scenarios, for example, searching for stocks with similar fluctuations~\cite{Rubio2017Improving},
looking for patients with similar EEG~\cite{Chaovalitwongse2006EEG}.

To mention about similarity search, we must start with similarity measure.
The most common similarity measures of time series are Euclidean distance and DTW distance.
Euclidean distance is parameter-free and has linear complexity,
but it is sensitive to shifting and scaling in the time axis~\cite{2017articleDTW}.
DTW is able to handle shifting and scaling
by searching an optimal match between the points of two sequences.
Besides, DTW can measure the similarity of time series with different lengths,
and achieve high matching precision~\cite{Jeong2011Weighted}.
In spite of the consideration of dozens of alternatives,
there is an increasing evidence that DTW distance is the best measure in most domains~\cite{Ding2008}.

However, the high computational complexity of DTW limits the applications of similarity search on large-scale datasets.
Therefore, research on efficient similarity search under DTW is of importance in both theory and practice.
In this paper, we propose a novel similarity search under DTW,
which can effectively improve search efficiency of time series with different lengths and guarantee no false dismissals.

The remainder of the paper is organized as follows.
In Section~\ref{sec:background}, we formulate the similarity search problem,
give a brief review of related work and state our motivation for this work.
In Section~\ref{sec:proposed method}, we propose a novel method of sequence extension,
and combine it with LB\_Keogh to form our lower bounding distance LB\_Keogh$^+$.
In Section~\ref{sec:similarity search}, we devise an efficient index of time series
and give the procedure of similarity search under DTW.
In Section~\ref{sec:experiments}, we conduct an extensive experiment to evaluate the validity of the proposed method.
Finally, we draw the conclusions and present future work in Section~\ref{sec:conclusions}.

\section{Background}\label{sec:background}

A time series is a continuous set of observations $x(t)$ ($t=1,2,\cdots ,n$),
arriving in time sequence~\cite{gorecki2018classification}.
Since time series from a dataset usually have the same time intervals,
they can be simply denoted as $X=(x_1,x_2,\cdots,x_n)$,
where $n$ is the length of the sequence.

Given two time series $X=(x_1,x_2,\cdots,x_n)$, $Y=(y_1,y_2,\cdots,y_m)$,
their DTW distance is defined as \cite{berndt1994using}:
\begin{equation}\label{eq:dtw definition}
\begin{split}
    & D_{dtw}(X,Y)= D_{base} (x_1 ,y_1 )+ \min \left\{ {\begin{array}{l}
 D_{dtw}(X,Y[2:-]) \\
 D_{dtw}(X[2:-],Y) \\
 D_{dtw}(X[2:-],Y[2:-]) \\
 \end{array}} \right.  \\
\end{split}
\end{equation}
where $D_{base}(x_1,y_1)$  is the base distance:
\begin{eqnarray}
&& D_{base}(x_i,y_j)=\left| {x_{i}-y_{j} } \right|
\end{eqnarray}

In the paper, we mainly focus on $\epsilon$-range search, which can be described as:
given a query sequence $Q$ and a time series dataset \{$C_1,C_2,\dots,C_k$\},
we need to retrieve all the sequences $C_i$ ($1 \leq i \leq k$) from the dataset,
such that $D_{dtw} (Q,C_i) \leq \epsilon$.

Since DTW has high computational complexity,
the retrieval efficiency of sequential scan is usually unacceptable.
Even worse, DTW does not satisfy triangular inequality,
making it difficult to devise appropriate index of time series under DTW distance.
Most methods devise their lower bounding distances to build index of time series.
Then, the index and lower bounding distances are used to speed up similarity search under DTW.
The key of these methods is their lower bounding distances.

There are three desirable properties of lower bounding distances~\cite{Nguyen2012Comparing}.
It does not incur false dismissals: lower bounding distances must be smaller than or equal to DTW distance.
It must be fast to compute: we would like their computational complexity to be linear in the length of the sequences.
It must be relatively tight: the lower bounding distance is close to DWT distance.

Among existing lower bounding distances, LB\_Kim~\cite{Kim2001An}, LB\_Yi~\cite{yi1998efficient}
and LB\_Keogh~\cite{keogh2005exact} are the most representative techniques.
They all have linear computational complexity and can guarantee no false dismissals.

LB\_Kim uses $L_\infty$ instead of $L_1$ or $L_2$ as its base distance.
Unlike most of other methods, the DTW distance defined by LB\_Kim is not the sum of the base distance.
Based on its defined DTW distance, LB\_Kim extracts four points (the first and last points, the maximum and minimum points) from each sequence,
and use these feature points to calculate the lower bounding distance.
LB\_Kim can be used for sequences with unequal length.

Given two time series, LB\_Yi chooses a sequence as the criterion.
Then, it extracts some points of the other sequence to calculate the lower bounding distance,
such that these extracted points are larger than the maximum value, or smaller than the minimum value of the criterion sequence.
LB\_Yi can also handle unequal-length sequences.

LB\_Keogh extracts the upper and lower boundary sequences from the query sequence $Q$,
by making use of global or local constraint of the warping path.
The area surrounded by the upper and lower boundary sequences is called envelope.
Based on these points of $C_i$ not falling into the envelope,
LB\_Keogh calculates the lower bounding distance between $Q$ and $C_i$.

LB\_Keogh has been verified to have higher tightness and pruning power than LB\_Kim and LB\_Yi.
Therefore, it has been recognized as the best one and attracts extensive attention.
Since then, Zhu~\cite{zhu2003warping}, Zhou~\cite{zhou2011boundary} and Li~\cite{Li2014Extensions}, \emph{etc}.
respectively proposed their improvements on LB\_Keogh to further enhance tightness and pruning power.

However, LB\_Keogh and its various improved methods can only deal with time series of the same length,
making it look slightly imperfect.
Just this defect alone is enough to seriously limit its practical application.
Because it is difficult to demandingly require the same length of time series in the real world.
In addition, the defect of LB\_Keogh forces DTW distancec to abandon its unique advantage
that it can measure unequal-length sequences.

From this perspective,
we try to propose a novel method of sequence extension to remove this defect of LB\_Keogh;
and then, we give the procedure of index building and similarity search;
finally, relevant proofs are completed to guarantee no false dismissals of our proposed method.
Basic notations in the work are summarized in Table \ref{tab:notations}.

\section{Proposed lower bounding distance of DTW}\label{sec:proposed method}

In this section, the calculation mechanism of DTW distance is analyzed.
Then, we propose a novel lower bounding distance LB\_Keogh$^+$,
which is a seamless combination of sequence extension and LB\_Keogh.
After that, we introduce several theorems and complete the relevant proofs
to guarantee no false dismissals of LB\_Keogh$^+$.

\begin{table}[H]
\centering
\caption{The notations in this work.}\label{tab:notations}
\begin{tabular}{p{1.2cm}<{\centering}|p{6.5cm}<{\centering}}
\hline
$D_{dtw}$             &  DTW distance \\ \hline
$D_{base}$            &  Base distance of DTW \\ \hline
$X[i:j]$              &  Subsequence of $X$ from the $i$-th to the $j$-th element  \\ \hline
$X^+$                 &  The extended sequence of $X$ by our proposed method  \\ \hline
$\bar{X}$             &  PAA form of a sequence $X$  \\ \hline
$Q$,\ $Q^+$           &  Query sequence and its extended sequence \\ \hline
$C$,\ $C^+$           &  Candidate sequence and its extended sequence \\ \hline
$M$                   &  Warping matrix \\ \hline
$W$                   &  Warping path \\ \hline
$w_k$                 &  The $k$-th element of a warping path \\ \hline
$U$,\ $L$             &  The upper and lower boundary sequences of $Q$  \\ \hline
$U^+$,\ $L^+$         &  The upper and lower boundary sequences of $Q^+$  \\ \hline
$Lmax$                &  The length after sequence extension  \\ \hline
$N$                   &  The number of sequences in a dataset \\ \hline
\end{tabular}
\end{table}


\subsection{The warping path of dynamic time warping}

To calculate DTW distance of $Q=(q_1,q_2,\cdots,q_n)$ and $C=(c_1,c_2,\cdots,c_m)$,
we usually construct a warping matrix $M \in \mathbb{R}^{n \times m}$~\cite{Salvador2007Toward},
where the element $M_{ij}$ is the base distance between $q_i$ and $c_j$.
A warping path $W$ is a set of elements in the warping matrix, as illustrated in Fig.~\ref{fig:the warping path}(c),
which is denoted as:
\begin{equation}\label{eq:warping path}
\begin{split}
    & W=w_1, w_2, \ldots, w_k, \ldots, w_K \\
    & max(n,m)\leq K < n+m-1.
\end{split}
\end{equation}
where $w_k=(i,j)_k$ denotes the matching relationship between $q_i$ and $c_j$.
Thus, a warping path defines a kind of mapping between all the points of $Q$ and $C$,
as illustrated in Fig.~\ref{fig:the warping path}(b)-(c).
Besides, a warping path is subject to the following constraints~\cite{Morel2017Time}.

Boundary condition: $w_1=(1, 1)_1$ and $w_K=(n, m)_K$.
An warping path should start from ($q_1$,$c_1$) and finish with ($q_n$,$c_m$).

Monotonicity: Given $w_k=(a, b)_k$ and $w_{k+1}=(a', b')_{k+1}$, we have $a'-a\geq 0$ and $b'-b\geq 0$.
This requires that a warping path must increase monotonously in time dimension.

Continuity: Given $w_k=(a, b)_k$ and $w_{k+1}=(a', b')_{k+1}$, we have $a'-a\leq 1$ and $b'-b\leq 1$.
This restricts the allowable steps in a warping path to adjacent elements (including diagonally adjacent elements).

\begin{figure}[htbp]
  \centering
  \includegraphics[width=0.48\textwidth]{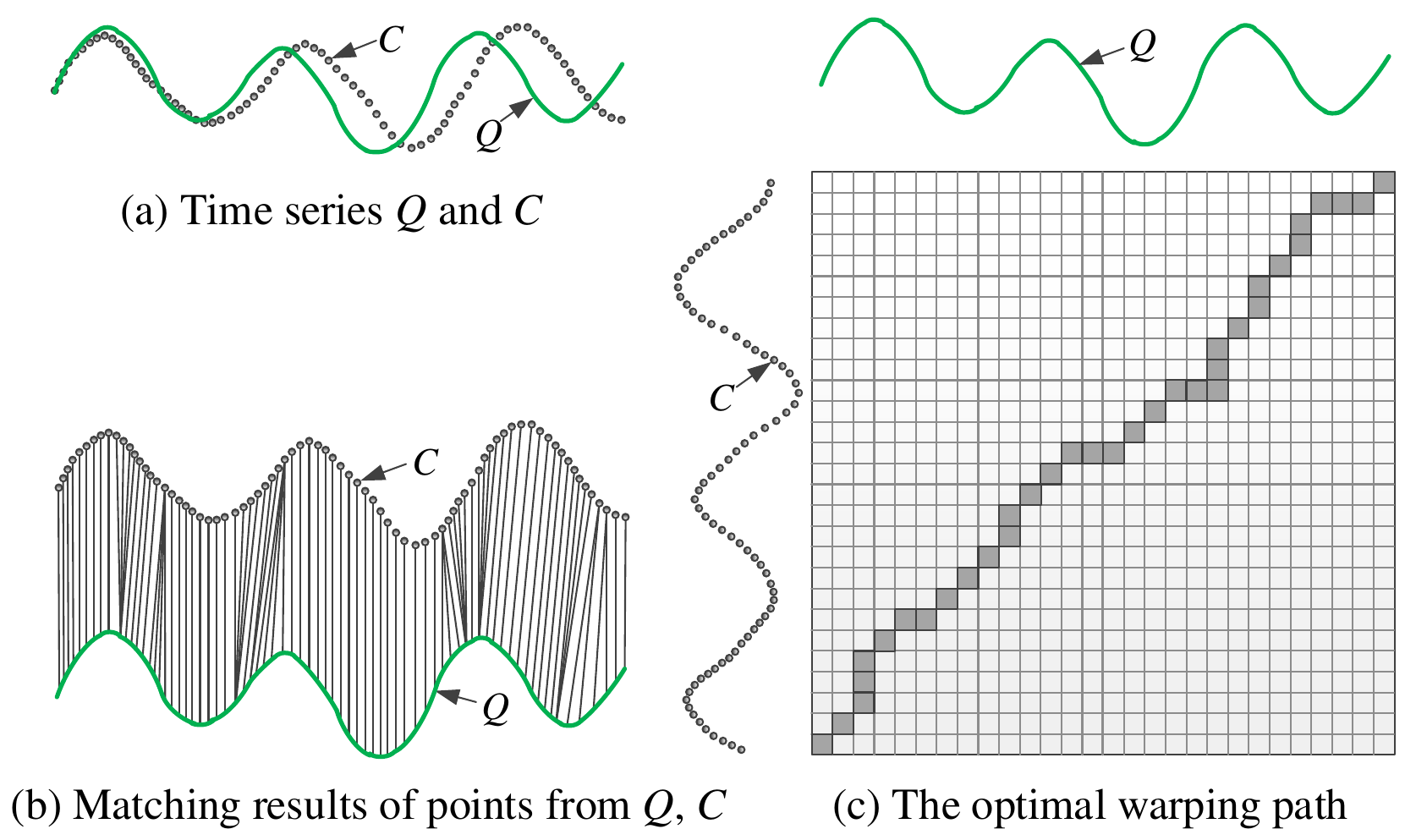}
  \caption{A warping path in the warping matrix~\cite{keogh2005exact}.}
  \label{fig:the warping path}
\end{figure}

For two sequences $Q$ and $C$,
there are several warping paths that satisfy the above three conditions.
DTW corresponds to the optimal path that minimizes the total cumulative distance:
\begin{equation}\label{eq:optimal warping path}
 D_{dtw}(Q,C) = min \{ \sum\nolimits_{k=1}^K {D_{base}{(w_k)}} \}.
\end{equation}
where $w_k=(i,j)_k$ and $D_{base}{(w_k)}$ is the base distance between $q_i$ and $c_j$.

\begin{theorem} \label{theorem:1}
  Given two sequences $Q$ and $C$,
  there is a warping path $W'=(w'_1,w'_2, \cdots,w'_{K'})$ in the warping matrix.
  If $\sum_{k=1}^{K'} D_{base}(w'_k)=\alpha$, then we have $D_{dtw}(Q,C)\leq\alpha$.
\end{theorem}

\begin{proof}
  In the warping matrix formed by $Q$ and $C$,
  we can find an optimal warping path $W=(w_1,w_2, \cdots,w_K)$.
  According to the definition of the optimal warping path,
  we can derive the following inequality:
\begin{equation}\label{eq:warping path}
    \sum_{k=1}^{K} D_{base}(w_k) \leq \sum_{k=1}^{K'} D_{base}(w'_k).
\end{equation}

  According to Eq.~\eqref{eq:optimal warping path}, we can further infer
  $D_{dtw}(Q,C) = \sum\nolimits_{k=1}^K {D_{base}{(w_k)}}$.
  Because $\sum_{k=1}^{K'} D_{base}(w'_k)=\alpha$, we can derive $D_{dtw}(Q,C)\leq\alpha$ holds.
\end{proof}

\subsection{Additional constraints on a warping path}

Some warping paths satisfying the above constraints may have pathological shape,
where a small section of one sequence maps onto a large section of another.
To avoid these pathological shapes,
global and local constraints are introduced to define the feasible scope of a warping path, called the warping window.

Sakoe-Chiba band and Itakura parallelogram are the most frequently used global constraints,
as illustrated in Fig.~\ref{fig:global constraints}.
Their warping windows are respectively a band and a parallelogram in the diagonal direction.
Local constraints define the permissible steps by the current position of a warping path.
In some cases, they can be reinterpreted as global constraints.
Here, we do not elaborate further.

\begin{figure}[!htbp]
\centering
\subfloat[Sakoe-Chiba band]{\includegraphics[width= 4.4 cm, height = 4.4 cm ]{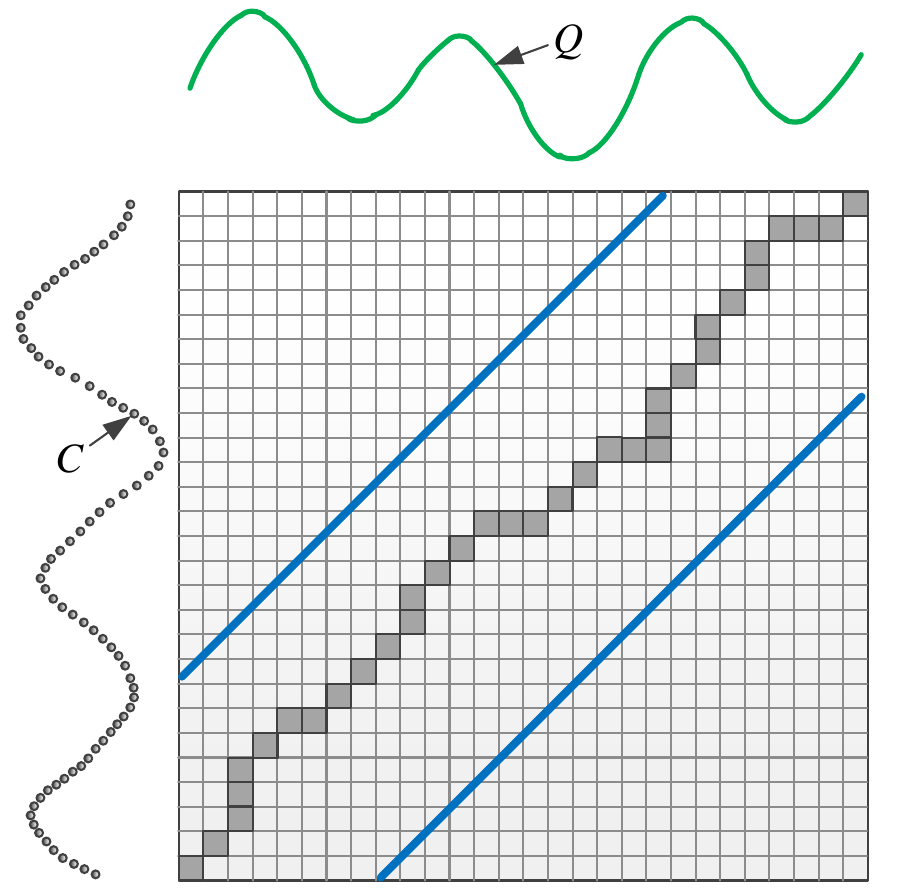}}
~~\subfloat[Itakura-Parallelogram]{\includegraphics[width= 4.4 cm, height = 4.4 cm ]{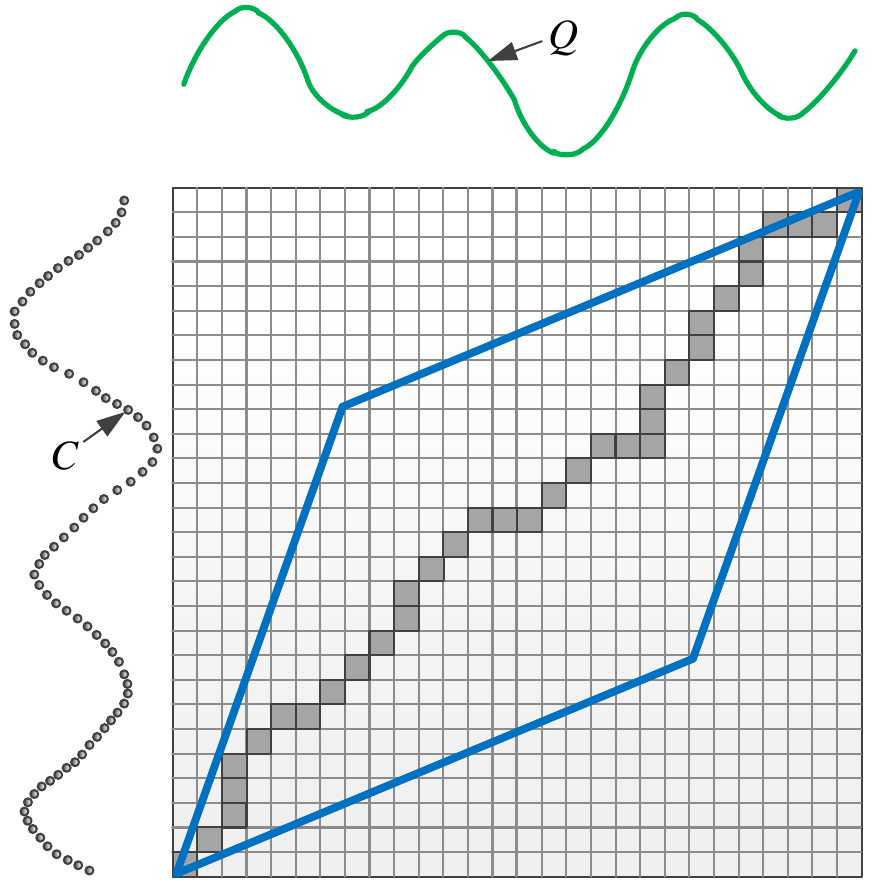}}
\caption{Global constraints on warping path~\cite{keogh2005exact}.}
\label{fig:global constraints}
\end{figure}

Without loss of generality, we mainly involve Sakoe-Chiba band in this paper,
which seems to be the most common constraint used in practice~\cite{Rabiner1978Considerations},~\cite{Dynamic1978programming} .
It constrains every element $w_k=(i, j)_k$ of a warping path such that $j-r\leq i \leq j+r$,
where $r$ is a constant specified the matching range of each point.
Thus, the difference of the lengths of any two sequences should satisfy the following theorem.

\begin{theorem} \label{theorem:2}
  Given two sequences $Q=(q_1, q_2, \cdots, q_n)$, $C=(c_1, c_2, \cdots, c_m)$,
  every warping path $W$ is constrained by Sakoe-Chiba band with the constant $r$.
  When we calculate $D_{dtw}(Q,C)$, the length difference between $Q$ and $C$ should not be greater than $r$.
\end{theorem}

\begin{proof}
  Reduction to absurdity is used to prove the theorem.

  We assume the length difference between $Q$ and $C$ is greater than $r$ when calculating $D_{dtw}(Q,C)$.
  That is, for the two sequences $Q=(q_1, q_2, \cdots, q_n)$ and $C=(c_1, c_2, \cdots, c_m)$,
  we have $n-m>r$ or $m-n>r$.

  According to the boundary condition of a warping path,
  $W$ should finish with ($q_n$,$c_m$),
  which means $q_n$ must match $c_m$.

  Because $W$ is also subject to the constraint of Sakoe-Chiba band.
  If $q_n$ matches $c_m$, we can infer $n-r\leq m \leq n+r$.
  Further, we have $n-m \leq r$ and $m-n \leq r$,
  which is just contradictory to the previous assumption.

  Therefore, we can deduce that the difference of the lengths of any two sequences should not be greater than $r$.
  \end{proof}

In the rest of this paper,
we suppose the length difference between a query sequence $Q$ and any candidate sequence $C$ should not be greater than $r$.

\begin{theorem} \label{theorem:3}
  Given two sequences $Q$ and $C$,
  under the constraints $r$ and $r^{'}$ of Sakoe-Chiba band,
  we respectively calculate their DTW distance $D_{dtw}(Q,C)$, $D_{dtw}^{'}(Q,C)$.
  If we have $r \leq r^{'}$,  then $D_{dtw}(Q,C) \geq D_{dtw}^{'}(Q,C)$ holds.
\end{theorem}

  \begin{proof}
  For two sequences $Q$ and $C$,
  we assume $W=(w_1,w_2, \cdots,w_K)$ is the optimal path under the constraint $r$ of Sakoe-Chiba band.
  According to Eq.~\eqref{eq:optimal warping path}, we have
  \begin{equation}\label{eq:r extension proof 1}
     D_{dtw}(Q,C) = \sum\nolimits_{k=1}^K {D_{base}{(w_k)}}
  \end{equation}

  When the width of the warping window expands from $r$ to $r^{'}$ ($r \leq r^{'}$),
  $W$ must be a feasible path in the warping matrix,
  as illustrated in Fig.~\ref{fig:global constraints}(a).

  That is, we can find a warping path $W$ under the constraint $r^{'}$.
  And we have deduced that Eq.~\eqref{eq:r extension proof 1} holds.
  According to Theorem~\ref{theorem:1}, we can infer Theorem~\ref{theorem:3} holds.
  \end{proof}

\subsection{Lower bounding distance}

Given a query sequence $Q=(q_1, q_2, \cdots, q_n)$,
its upper and lower boundary sequences $U = (u_1, u_2,\ldots, u_n)$ and $L = (l_1, l_2,\ldots, l_n)$ are respectively defined as:
\begin{equation}\label{eq:envelope}
\begin{split}
    & u_i = max\{q_{i-r}:q_{i+r} \} \\
    & l_i = min\{q_{i-r}:q_{i+r} \} \\
    & i-r\geqslant 1, \quad  i+r\leq n
\end{split}
\end{equation}
where $r$ is a constant that comes with the Sakoe-Chiba band.
The query sequence $Q$ is enclosed in the region formed by $U$ and $L$,
as illustrated in Fig.~\ref{fig:upper and lower boundary}.
The region between $U$ and $L$ is called envelope.

\begin{figure}[htbp]
  \centering
  \includegraphics[width=0.42\textwidth]{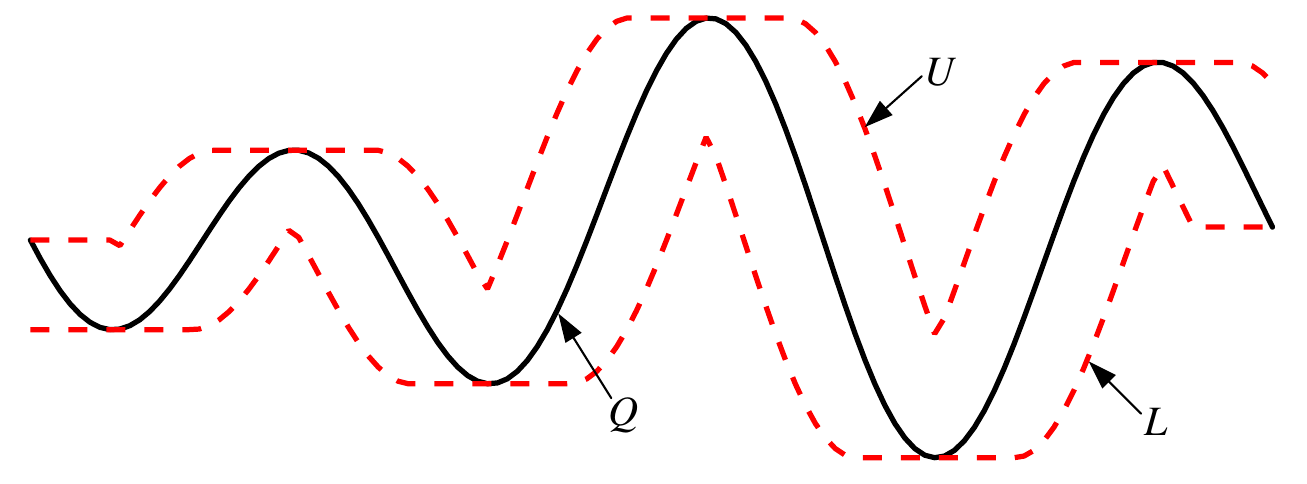}
  \caption{A query sequence and its upper and lower boundary sequences.}
  \label{fig:upper and lower boundary}
\end{figure}

For any sequence $C=(c_1, c_2, \cdots, c_m)$ in a dataset,
only if $m=n$, the lower bounding distance of DTW between $Q$ and $C$ can be defined as~\cite{keogh2005exact}:
\begin{equation}\label{eq:lb dtw}
\begin{split}
    & LB\_Keogh(Q,C)= \sum\limits_{i=1}^{n}  \left\{ {\begin{array}{l}
 |c_i-u_i|,  \ if \ c_i>u_i \\
 |c_i-l_i|,  \ if \  c_i<l_i \\
  0,  \quad otherwise \\
 \end{array}} \right.  \\
\end{split}
\end{equation}

LB\_Keogh$(Q,C)$ calculates the sum of the base distance
between any point of $C$ not falling into the envelope
and the corresponding point in the nearest boundary sequence.
Fig.~\ref{fig:the lower bounding distance} visualizes the meaning of LB\_Keogh$(Q,C)$,
where the shadow areas represent the parts that need to be cumulated by the base distance.

\begin{figure}[htbp]
  \centering
  \includegraphics[width=0.43\textwidth]{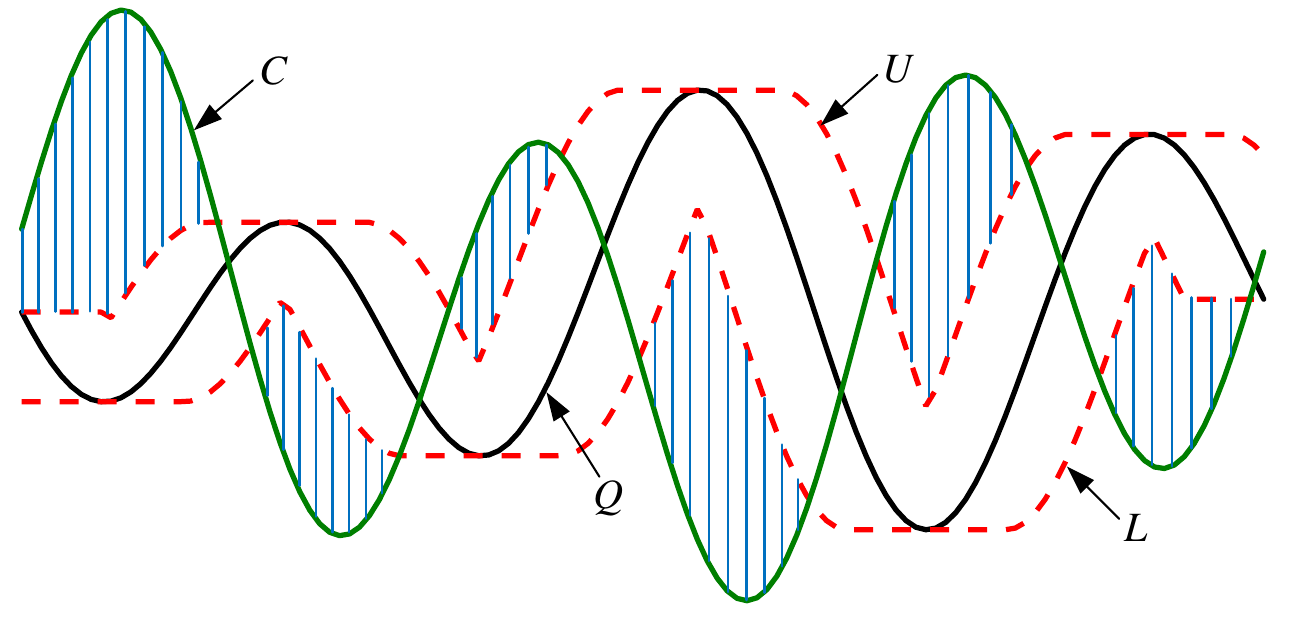}
  \caption{A visual intuition of LB\_Keogh$(Q,C)$}
  \label{fig:the lower bounding distance}
\end{figure}

In addition, LB\_Keogh satisfies the following theorem. Please refer to~\cite{keogh2005exact} for detailed proof.
\begin{theorem}\label{theorem:4}
For two sequences $Q$ and $C$ of the same length,
if every warping path is constrained by Sakoe-Chiba band with a constant $r$,
the inequality LB\_Keogh(Q,C)$\leq D_{dtw}(Q,C)$ holds.
\end{theorem}

When LB\_Keogh is used for similarity search, Theorem~\ref{theorem:4} can ensure that there is no false dismissals.
Besides, it has low computational complexity and good tightness.
Therefore, LB\_Keogh is so far the best lower bounding distance of DTW.
However, it can only measure sequences of the same length,
severely limiting its application in practical scenarios.

\subsection{Technique of sequence extension}

In order to solve the intractable problem above,
we propose a novel method of sequence extension.
Given two sequences $Q=(q_1, q_2, \cdots, q_n)$, $C=(c_1, c_2, \cdots, c_m)$,
we add the same arbitrary constant $e$ after them to construct two sequences of the same length:
\begin{equation}\label{eq:sequence extension}
\begin{split}
    & Q^+=(q_1, q_2, \cdots, q_n, q_{n+1},\cdots,q_{Lmax}) \\
    & C^+=(c_1, c_2, \cdots, c_m, c_{m+1}, \cdots,c_{Lmax}) \\
    & q_{n+1}=\cdots=q_{Lmax}=e \\
    & c_{m+1}=\cdots=c_{Lmax}=e \\
    & Lmax>max(n,m)
\end{split}
\end{equation}
where the constraint $Lmax>max(n,m)$ means that we should add at least one element to the longer sequence.

From the external form,
the equal-length sequences $Q^+$ and $C^+$ can be directly used in LB\_Keogh.
More important,
we need to further analyze whether the proposed method can guarantee no false dismissals,
when $D_{dtw}(Q^+,C^+)$ is used for similarity search.

\begin{theorem} \label{theorem:5}
  Given two sequences $Q=(q_1, q_2, \cdots, q_n)$, $C=(c_1, c_2, \cdots, c_m)$,
  they are extended to the equal-length sequences $Q^+$, $C^+$ by Eq.~\eqref{eq:sequence extension}.
  Then, the following inequality holds:
  \begin{equation}\label{eq:path extension theorem 1}
    D_{dtw}(Q^+,C^+)\leq D_{dtw}(Q,C)
  \end{equation}
\end{theorem}

\begin{proof}
For the original sequences $Q$ and $C$,
we assume that $W=(w_1,w_2, \cdots,w_K)$ is the optimal warping path,
as illustrated in Fig.~\ref{fig:equal-length extension}(a).
According to Eq.~\eqref{eq:optimal warping path}, we have
\begin{equation}\label{eq:path extension proof 1}
    D_{dtw}(Q,C) = \sum\nolimits_{k=1}^K {D_{base}{(w_k)}}
\end{equation}

For the extended equal-length sequences $Q^+$ and $C^+$,
we can simply construct a warping path $W^+$,
as illustrated in Fig.~\ref{fig:equal-length extension}(b).
\begin{equation}\label{eq:path extension proof 2}
    W^+=(w_1,w_2,\cdots,w_K,w_{K+1},\cdots,w_{K+p})
\end{equation}

The front part of $W^+$ is exactly the same as $W$,
contained in the red rectangle of Fig.~\ref{fig:equal-length extension}(b).
The back part of $W^+$ is formed by $(q_{n+1},\cdots,q_{Lmax})$ and $(c_{m+1},\cdots,c_{Lmax})$,
which are the extended parts of $Q^+$ and $C^+$,
as shown in the blue rectangle of Fig.~\ref{fig:equal-length extension}(b).

\begin{figure}[!htbp]
\centering
\subfloat[The optimal path of $Q$ and $C$]{\includegraphics[width= 3.2 cm, height = 4 cm ]{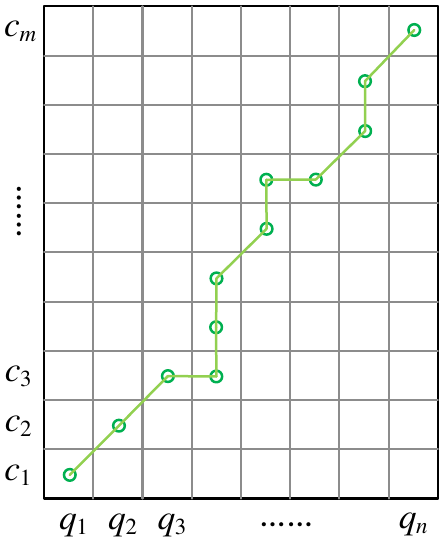}}
\subfloat[A warping path of $Q^+$ and $C^+$]{\includegraphics[width= 5.4 cm, height = 5.4 cm ]{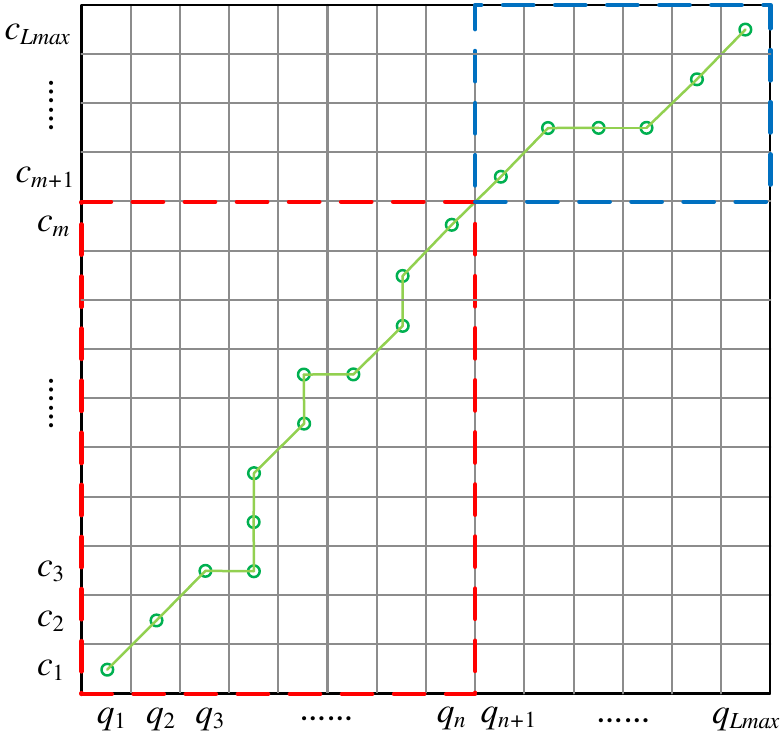}}
\caption{A warping path of the extended equal-length sequences.}
\label{fig:equal-length extension}
\end{figure}

Because $q_{n+1}=\cdots=q_{Lmax}=c_{m+1}=\cdots=c_{Lmax}=e$,
from $W^+$ we can infer:
\begin{equation}\label{eq:path extension proof 3}
    \sum_{k=K+1}^{K+p} D_{base}(w_k)=0
\end{equation}

For the whole warping path $W^+$, we have:
\begin{equation}\label{eq:path extension proof 4}
    \sum_{k=1}^{K+p} D_{base}(w_k)=\sum_{k=1}^{K} D_{base}(w_k)+\sum_{k=K+1}^{K+p} D_{base}(w_k)
\end{equation}

We substitute Eq.~\eqref{eq:path extension proof 1} and Eq.~\eqref{eq:path extension proof 3} into Eq.~\eqref{eq:path extension proof 4}:
\begin{equation}\label{eq:path extension proof 5}
    \sum_{k=1}^{K+p} D_{base}(w_k)=D_{dtw}(Q,C)
\end{equation}

That is, we can find a warping path $W^+$ in the warping matrix formed by $Q^+$ and $C^+$.
The total cumulative distance of $W^+$ is $D_{dtw}(Q,C)$.
According to Theorem~\ref{theorem:1},
we can deduce that Eq.~\eqref{eq:path extension theorem 1} holds.
\end{proof}

From Theorem~\ref{theorem:4},
we can infer the proposed method of sequence extension can guarantee no false dismissals.
Under the constraint $Lmax > max(n,m)$ in Eq.~\eqref{eq:sequence extension},
we need to further discuss another important problem:
how does the extended length affect similarity search?

\begin{theorem} \label{theorem:6}
  Given two sequences $Q=(q_1, q_2, \cdots, q_n)$, $C=(c_1, c_2, \cdots, c_m)$,
  they are extended to the equal-length sequences $Q^+$, $C^+$ by Eq.~\eqref{eq:sequence extension}.
  Then we further add the same constant $e$ after $Q^+$ and $C^+$,
  to obtain two new sequences of the same length:
  \begin{equation}\label{eq:sequence extension further}
  \begin{split}
    & Q^{++}=(q_1, q_2, \cdots, q_{Lmax}, q_{Lmax+1}, \cdots,q_{Lmax+t}) \\
    & C^{++}=(c_1, c_2, \cdots, c_{Lmax}, c_{Lmax+1}, \cdots,c_{Lmax+t}) \\
    & q_{Lmax}=q_{Lmax+1}=\cdots=q_{Lmax+t}=e \\
    & c_{Lmax}=c_{Lmax+1}=\cdots=c_{Lmax+t}=e, \quad t\geq 1 \\
  \end{split}
  \end{equation}

  Then, the following two inequalities hold:
  \begin{equation}\label{eq:path extension further theorem 1}
   D_{dtw}(Q^{++},C^{++})= D_{dtw}(Q^+,C^+)
  \end{equation}
  \begin{equation}\label{eq:path extension further theorem 2}
   LB\_Keogh(Q^{++},C^{++})=LB\_Keogh(Q^+,C^+)
  \end{equation}
\end{theorem}


\begin{proof}
For the original sequences $Q$ and $C$,
we assume their extended equal-length sequences are denoted as
$Q^+=(q_1, \cdots, q_n, q_{n+1},\cdots,q_{Lmax})$,
$C^+=(c_1, \cdots, c_m, c_{m+1}, \cdots,c_{Lmax})$.

Because $Lmax>max(n,m)$,
$Q^+$ and $C^+$ should add at least one element to the end of $Q$ and $C$.
Without loss of generality, we assume $m \geq n$
and add one element $e$ behind $C$ to get $C^+$.
Correspondingly, we add $m-n+1$ elements behind $Q$ to get $Q^+$.

Then we can get the optimal matching relations of all points from $Q^+$ and $C^+$,
as illustrated in Fig.~\ref{fig:extended proof}(a),
which correspond to the optimal warping path.
We get $D_{dtw}(Q^+,C^+)$ by calculating the sum of the base distance between all those matching points.

Further, we add the same number of elements behind $Q^+$ and $C^+$,
to form another extended equal-length sequences $Q^{++}$, $C^{++}$.
It's worth noting that $t$ can be any positive integer under the constraint $t\geq 1$,
which means that the length of sequence extension in Eq.~\eqref{eq:sequence extension further} is arbitrary.

We can find the optimal matching relations of all points from $Q^{++}$ and $C^{++}$
in the process of calculating $D_{dtw}(Q^{++},C^{++})$.
To simplify the analysis,
we divide $Q^{++}$ and $C^{++}$ into three contiguous segments,
as illustrated in Fig.~\ref{fig:extended proof}(b).

\begin{figure}[!htbp]
\centering
\subfloat[The first extended sequences $Q^+$, $C^+$]{\includegraphics[width= 6 cm, height = 2.5 cm ]{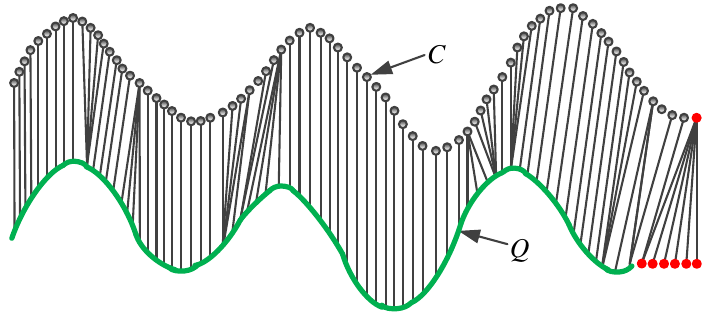}}\\
\subfloat[The second extended sequences $Q^{++}$, $C^{++}$]{\includegraphics[width= 7 cm, height = 3.3 cm ]{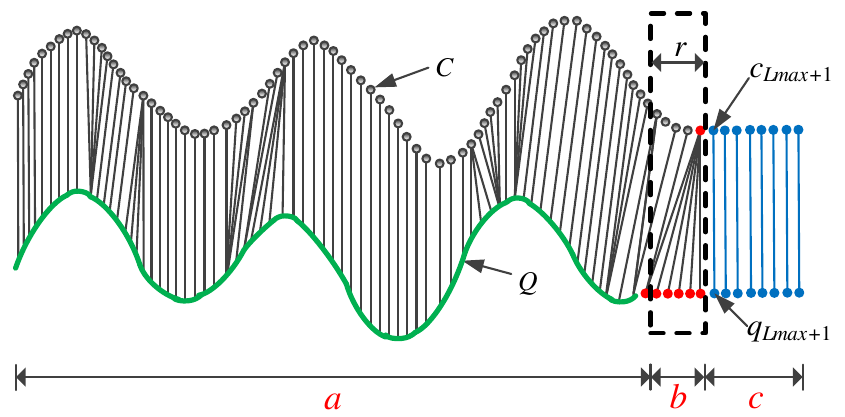}}\\
\caption{The calculation of DTW distance for extended sequences.}
\label{fig:extended proof}
\end{figure}

\textbf{Segment a}.
This segment is formed by $Q^{++}[1:Lmax-r]$, $C^{++}[1:Lmax-r]$.
When calculating $D_{dtw}(Q^{++},C^{++})$,
due to the constraint $r$ of Sakoe-Chiba band,
all the points from $Q^{++}$ and $C^{++}$ in this segment will not be affected by the extended sequences
$(q_{Lmax+1}, \cdots,q_{Lmax+t})$, $(c_{Lmax+1}, \cdots,c_{Lmax+t})$.
For all the points of $Q^{++}$ and $C^{++}$ in this segment,
their optimal matching relations are the same as those of $Q^+$ and $C^+$.

\textbf{Segment b}.
This segment is formed by $Q^{++}[Lmax-r+1:Lmax]$, $C^{++}[Lmax-r+1:Lmax]$.
The points from $Q^{++}$ and $C^{++}$ in this segment have the possibility to be affected by the extended sequences
$(q_{Lmax+1}, \cdots,q_{Lmax+t})$, $(c_{Lmax+1}, \cdots,c_{Lmax+t})$.
In the process of calculating DTW, any point on a sequence can replicate themselves.
Because $q_{Lmax}=q_{Lmax+1}=\cdots=q_{Lmax+t}$ and $c_{Lmax}=c_{Lmax+1}=\cdots=c_{Lmax+t}$.
The extended sequences $(q_{Lmax+1}, \cdots,q_{Lmax+t})$, $(c_{Lmax+1}, \cdots,c_{Lmax+t})$ can be understood as
self-replications of $q_{Lmax}$ and $c_{Lmax}$.
For all the points of $Q^{++}$ and $C^{++}$ in this segment,
the sum of the base distances of these points is equal to the counterpart of $Q^+$ and $C^+$.

\textbf{Segment c}.
This segment is formed by $Q^{++}[Lmax+1:Lmax+t]$, $C^{++}[Lmax+1:Lmax+t]$.
Because $q_{Lmax+1}=\cdots=q_{Lmax+t}=c_{Lmax+1}=\cdots=c_{Lmax+t}$.
For all the points of $Q^{++}$ and $C^{++}$ in this segment,
the sum of the base distance of these points is equal to 0.

Because $D_{dtw}(Q^{++},C^{++})$ is calculated by the sum of the base distance between all those matching points,
corresponding to the optimal warping path.
To sum up the above analysis, we can infer Eq.~\eqref{eq:path extension further theorem 1} holds.

We assume the upper and lower boundary sequences of $Q^+$ are denoted as:
\begin{equation}\label{eq:path extension further proof 1}
\begin{split}
   & U{^+} = ( u_1, u_2,\ldots, u_{Lmax} ) \\
   & L{^+} = ( l_1, l_2,\ldots, l_{Lmax} )
\end{split}
\end{equation}
where $u_i$, $l_i$ are defined in Eq.~\eqref{eq:envelope}.
Then we have
\begin{equation}\label{eq:path extension further proof 2}
\begin{split}
& LB\_Keogh(Q^+,C^+)= \sum\limits_{i=1}^{Lmax}  \left\{ {\begin{array}{l}
 |c_i-u_i|,  \ if \  c_i>u_i \\
 |c_i-l_i|,  \ if \  c_i<l_i \\
  0,  \quad otherwise \\
 \end{array}} \right.  \\
\end{split}
\end{equation}

Similarly, the upper and lower boundary sequences of $Q^{++}$ are denoted as $U{^{++}}$, $L{^{++}}$.
According to the definition of $Q^{++}$ in Eq.~\eqref{eq:sequence extension further},
we can derive that the first $Lmax$ elements of $U{^{++}}$, $L{^{++}}$ are just $U{^+}$ and $L{^+}$:
\begin{equation}\label{eq:path extension further proof 3}
\begin{split}
   & U{^{++}} = (u_1,\ldots, u_{Lmax}, u_{Lmax+1}, \ldots, u_{Lmax+t} ) \\
   & L{^{++}} = (l_1,\ldots, l_{Lmax}, l_{Lmax+1}, \ldots, l_{Lmax+t} )
\end{split}
\end{equation}

According to the definition of LB\_Keogh, we have
\begin{equation}\label{eq:path extension further proof 4}
\begin{split}
    & LB\_Keogh(Q^{++},C^{++})= \sum\limits_{i=1}^{Lmax+t}  \left\{ {\begin{array}{l}
 |c_i-u_i|,  \ if \  c_i>u_i \\
 |c_i-l_i|,  \ if \  c_i<l_i \\
  0,  \quad otherwise \\
 \end{array}} \right.  \\
\end{split}
\end{equation}

We substitute Eq.~\eqref{eq:path extension further proof 2} into Eq.~\eqref{eq:path extension further proof 4}:
\begin{equation}\label{eq:path extension further proof 5}
\begin{split}
    & LB\_Keogh(Q^{++},C^{++})= LB\_Keogh(Q^{+},C^{+}) + \\
    & \quad\quad\quad\quad\quad\quad \sum\limits_{i=Lmax+1}^{Lmax+t}  \left\{ {\begin{array}{l}
 |c_i-u_i|,  \ if \  c_i>u_i \\
 |c_i-l_i|,  \ if \  c_i<l_i \\
  0,  \quad otherwise \\
 \end{array}} \right.  \\
\end{split}
\end{equation}

Because $q_{Lmax+1}=\cdots=q_{Lmax+t}=c_{Lmax+1}=\cdots=c_{Lmax+t}$,
according to the definition of upper and lower boundary sequences in Eq.~\eqref{eq:envelope},
we have
\begin{equation}\label{eq:path extension further proof 6}
\begin{split}
    & \sum\limits_{i=Lmax+1}^{Lmax+t}  \left\{ {\begin{array}{l}
      |c_i-u_i|,  \ if \  c_i>u_i \\
      |c_i-l_i|,  \ if \  c_i<l_i \\
      0,  \quad otherwise \\
 \end{array}}  = 0 \right.  \\
\end{split}
\end{equation}


Therefore, we can deduce that Eq.~\eqref{eq:path extension further theorem 2} holds.

Similarly, we can prove that Theorem.~\ref{theorem:6} still holds
when the constraint $r$ of Sakoe-Chiba band is removed.
\end{proof}

From Theorem~\ref{theorem:6},
we can get an important property of the proposed method of sequence extension.
If the original sequences $Q$, $C$ are extended to $Q^+$, $C^+$ by Eq.~\eqref{eq:sequence extension},
the extended length is independent of $D_{dtw}(Q^+,C^+)$ and LB\_Keogh$(Q^+,C^+)$.

\section{Similarity search under dynamic time warping}\label{sec:similarity search}

In order to improve the efficiency of similarity search, it is necessary to resort to indexes.
If time series are directly organized by indexes,
the performance of similarity search will seriously degrade.
In this section, we reduce the dimension of time series,
and organize them with spatial indexes.
Then, we present the procedure of $\epsilon$-range search.

\subsection{Index of time series}\label{sec:index design}

In a dataset, any candidate sequence $C=(c_1, c_2, \cdots, c_m)$ can be represented by a $N$-dimensional vector $\bar{C} = (\bar{c}_1, \bar{c}_2,\ldots,\bar{c}_N)$,
where the $i$-th element $\bar{c}_i$ is defined as:
\begin{equation}\label{eq:PAA}
 \bar{c}_i= \frac{N}{n} \sum\limits_{k=\frac{n}{N}(i-1)+1}^{\frac{n}{N}i} c_k
\end{equation}

The transition from $C$ to $\bar{C}$ is called piecewise aggregate approximation (PAA).
In practical applications, $n$ cannot be exactly an integer multiple of $N$.
We can simply solve the problem by using our sequence extension method.
Detailed analysis will be carried out in the section of experiments.

Given two sequences $Q$, $C$,
we transform them into equal-length sequences
$Q^+=(q_1, q_2, \cdots, q_n, q_{n+1},\cdots,q_{Lmax})$, $C^+=(c_1, c_2, \cdots, c_m, c_{m+1}, \cdots,c_{Lmax})$.
Then we have
\begin{equation}\label{eq:extented paa}
\begin{split}
   & \bar{Q^+}=(\bar{q^+_1}, \bar{q^+_2},\ldots,\bar{q^+_N}) \\
   & \bar{C^+}=(\bar{c^+_1}, \bar{c^+_2},\ldots,\bar{c^+_N})
\end{split}
\end{equation}

The distance between $\bar{Q^+}$ and $\bar{C^+}$ is defined as:
\begin{equation}\label{eq:PAA distance}
 D_{PAA}(\bar{Q^+},\bar{C^+})=
 \frac{Lmax}{N} \sum\limits_{i=1}^{N} \left| \bar{q^+_i}-\bar{c^+_i} \right|
\end{equation}

We have the following inequality. Please refer to~\cite{Keogh2001Dimensionality},~\cite{Yi2000Fast} for detailed proof.
\begin{equation}\label{eq:PAA distance proof}
 D_{PAA}(\bar{Q^+},\bar{C^+}) \leq
  \sum\limits_{i=1}^{Lmax} \left| q^+_i-c^+_i \right|
\end{equation}

For any candidate sequence $C$ in a dataset,
we extend it to the sequence $C^+$ of length $Lmax$ by our proposed method.
Then, $C^+$ is transformed into $N$-dimensional vector $\bar{C^+}$ by PAA.
Thus, we can make use of spatial indexes to organize these $N$-dimensional vectors.

Without loss of generality,
we use R-Tree to organize the PAA form of candidate sequences.
Supposing $V$ is a leaf node of R-Tree,
the MBR (Minimum Bounding Rectangle) related to the leaf node $V$ is denoted as $R=(B,H)$,
where $B=(b_1,b_2,\ldots,b_N)$, $H=(h_1,h_2,\ldots,h_N)$
are the lower and upper boundaries of MBR.
Any candidate sequence,
the PAA form of which is contained in the MBR,
will be included in the leaf node $V$.

\subsection{The procedure of similarity search}\label{sec:similarity search procedure}

We extract the feature of a query sequence and take it as the input of $\epsilon$-range search.
For a query sequence $Q$,
we extend it to the sequence $Q^+$ of length $Lmax$ by our proposed method.
Next, we obtain the upper and lower boundary sequences of $Q^+$,
denoted as $U^+$ and $L^+$.
Then, $U^+$, $L^+$ are transformed into $N$-dimensional vectors
$\bar{U^+}=(\bar{u^+_1}, \bar{u^+_2},\ldots,\bar{u^+_N})$, $\bar{L^+}=(\bar{l^+_1}, \bar{l^+_2},\ldots,\bar{l^+_N})$ by PAA.

In order to complete similarity search,
we need to introduce another two lower bounding distances: LB\_PAA and LB\_MBR.

LB\_PAA between $Q^+$ and $C^+$ is defined as:
\begin{equation}\label{eq:LBPAA}
\begin{split}
	LB\_{PAA}(\bar{U^+},\bar{L^+},\bar{C^+})=
	\frac{n}{N} \sum\limits_{i=1}^{N}  \left\{ {\begin{array}{l}
    |\bar{c^+_i} - \bar{u^+_i}|, \ if \  \bar{c^+_i}>\bar{u^+_i}\\
    |\bar{c^+_i}-\bar{l^+_i}|,   \ if \  \bar{c^+_i}<\bar{l^+_i} \\
    0,  \quad \quad \quad \quad otherwise \\
 \end{array}} \right.
\end{split}
\end{equation}

According to  Eq.~\eqref{eq:PAA distance proof}, we have
\begin{equation}\label{eq:lower boundary of Keogh PAA}
 LB\_{PAA}(\bar{U^+},\bar{L^+},\bar{C^+}) \leq LB\_Keogh(Q^+,C^+)
\end{equation}

LB\_MBR between $Q^+$ and MBR $R$ is defined as:
\begin{equation}\label{eq:LBPAAextend}
\begin{split}
	LB\_MBR(\bar{U^+},\bar{L^+},R)=
	\frac{n}{N} \sum\limits_{i=1}^{N}  \left\{ {\begin{array}{l}
    |\bar{u^+_i} - h_i|, \ if \  \bar{u^+_i}>h_i\\
    |b_i-\bar{l^+_i}|,   \ if \  \bar{l^+_i}<b_i \\
    0,  \quad \quad \quad \quad otherwise \\
 \end{array}} \right.
\end{split}
\end{equation}

Given the extended query sequence $Q^+$ and MBR $R$,
for any extended candidate sequence $C^+$ contained in $R$,
we have the following inequality:
\begin{equation}\label{eq:DistMBR}
\begin{split}
	LB\_MBR(\bar{U^+},\bar{L^+},R) \leq  LB\_{PAA}(\bar{U^+},\bar{L^+},\bar{C^+})
\end{split}
\end{equation}

For detailed proof, please refer to~\cite{keogh2005exact}.
The procedure of $\epsilon$-range search is summarized in Algorithm~\ref{algorithm:1}.
\begin{algorithm}[h]
\caption{ RangeSearch($Q,P,\epsilon$)}\label{algorithm:1}
\begin{algorithmic}[1]
\REQUIRE A query sequence $Q$, time series dataset $C_i^+$ ($1 \leq i \leq k$),
the root node $P$ of R-tree, distance threshold $\epsilon$.
\ENSURE Result set \textbf{R} of $\epsilon$-range search. \\
\STATE Initialize $Q^+$, $U^+$, $L^+$, $\bar{U^+}$, $\bar{L^+}$; \\
\IF { $P$ is a non-leaf node}
\FOR {each child node $T$ of $P$}
\IF {LB\_MBR$(\bar{U^+},\bar{L^+},R) \leq \epsilon$}
\STATE //$R$ is the MBR corresponding to node $T$.
\STATE RangeSearch($Q,T,\epsilon$);
\ENDIF
\ENDFOR
\ELSE
\FOR {each PAA point $\bar{C^+_i}$ in $P$}
\IF {LB\_{PAA}$(\bar{U^+},\bar{L^+},\bar{C^+_i}) \leq \epsilon$}
\STATE retrieve original sequence $C_i$ from the dataset;
\ENDIF
\IF {$D_{dtw}(Q,C_i) \leq \epsilon$}
\STATE add $C_i$ to \textbf{R};
\ENDIF
\ENDFOR
\ENDIF
\end{algorithmic}
\end{algorithm}

\subsection{Analysis of effectiveness and complexity}

So far, we have introduced many kinds of lower bounding distances.
LB\_PAA and LB\_MBR are used for similarity search on spatial indexes,
and LB\_Keogh is the basis for the two methods.
Fig.~\ref{fig:noFalseDismissal} illustrates the relationship between different distances.
For two original sequences $Q$, $C$,
if $D_{dtw}(Q,C) \leq \epsilon$,
we can get $D_{dtw}(Q^+,C^+)$ and these lower bounding distances
are all less than or equal to $\epsilon$.
Therefore, our proposed similarity search can guarantee no false dismissals.

\begin{figure}[htbp]
  \centering
  \includegraphics[width=0.47\textwidth]{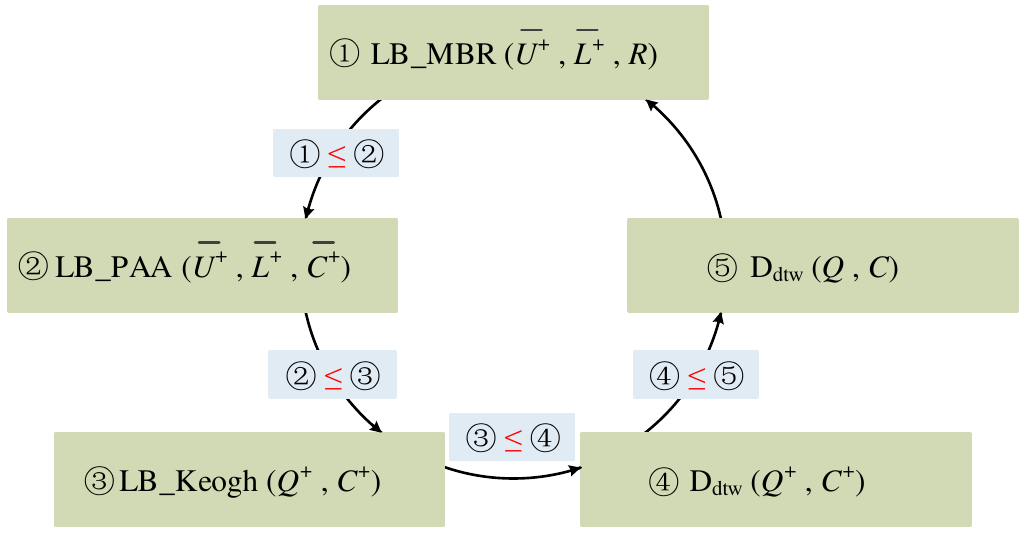}
  \caption{The relationship between different distances.}
  \label{fig:noFalseDismissal}
\end{figure}

The complexity of DTW distance between $Q$ and $C$ is $O(nm)$,
where $n$, $m$ are the lengths of $Q$ and $C$ respectively.
In our proposed method,
the computational cost of sequence extension is very low and even can be ignored.
The complexity of LB\_Keogh, LB\_{PAA} and LB\_MBR is at most linear in the length of sequences.
Compared with sequential scan under DTW distance,
our similarity search can effectively improve the retrieval efficiency.

\section{Experimental evaluation}\label{sec:experiments}

Our experiments are carried out on a PC with Intel Core i7-8550U CPU and 16 GB RAM,
running with Matlab R2018a.
The proposed method is evaluated on 10 benchmark datasets coming from the UCR time series repository~\cite{UCRDataSet},
where the lengths of the sequences on a dataset are all the same.
To obtain sequences of different lengths, we truncate a random length at the end of every sequence,
such that any pair of the sequences on a dataset satisfies Theorem~\ref{theorem:2}.

These datasets cover a wide range of applications, such as energy, medicine, image matching, motion recognition, \emph{etc}.
The average lengths of the sequences in these datasets vary from 23 to 683.
More information of the datasets is shown in Table~\ref{tab:datasets}.

In our experiments,
the constraint $r$ of Sakoe-Chiba band is equal to 10\% of the length of the longest original sequence on a dataset.
Because this value appears to be the most commonly used in many literatures.

\subsection{The evaluation of validity for sequence extension}

Due to space limitation,
we just choose GunPoint dataset to evaluate the validity of sequence extension.
It involves one female actor and one male actor making a motion with their hand,
and contains two classes: Gun-Draw and Point,
as illustrated in Fig.~\ref{fig:GunPoint}.

\begin{table}[H]
\centering
\caption{The details of the benchmark datasets.}\label{tab:datasets}
\begin{tabular}{p{0.3cm}<{\centering}p{2cm}p{1.2cm}<{\centering}p{1cm}<{\centering}p{0.8cm}<{\centering}p{0.9cm}<{\centering}}
\hline
\textbf{ID} & \textbf{Dataset} & \textbf{Length} & \textbf{Av.length} &	\textbf{Class} &	\textbf{Sample} \\
\hline
1  & ItalyPowerDemand  &22--24 &23 &2 &1096 \\
2  & SyntheticControl  &54--60 &57 &6 &600 \\
3  & ECG5000  &126--140 &133 &5 &5000 \\
4  & GunPoint  &135--150 &143 &2 &200 \\
5  & WordSynonyms  &243--270 &256 &25 &905 \\
6  & Words50  &243--270 &257 &50 &905 \\
7 & Symbols  &359--398 &379 &6  &1020 \\
8 & Yoga  &384--426 &405 &2 &3300 \\
9 & ShapesAll  &461--512 &486 &60 &1200 \\
10 & Computers  &648--720 &683 &2 &500 \\
\hline
\end{tabular}
\end{table}

For Gun-Draw the actors have their hands by their sides.
They draw a replicate gun from a hip-mounted holster, point it at a target for approximately one second,
then return the gun to the holster, and their hands to their sides.
For Point the actors have their gun by their sides.
They point with their index fingers to a target for approximately one second, and then return their hands to their sides.
For both classes, we record the centroid of the actor's right hands in X-axis, which appear to be highly correlated.

\begin{figure}[htbp]
  \centering
  \includegraphics[width=0.36\textwidth]{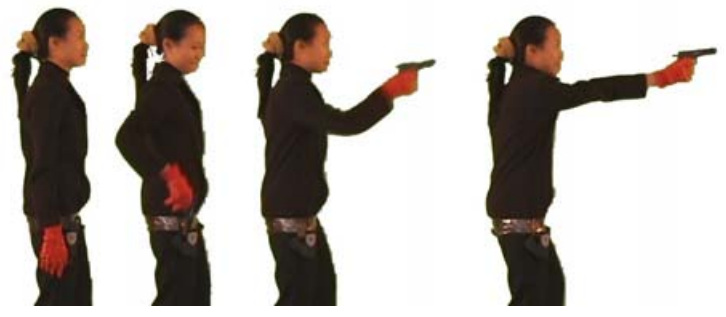}
  \caption{Description of GunPoint dataset.}
  \label{fig:GunPoint}
\end{figure}

We extend all the sequences to the same minimum length $Lmax=151$, according to Eq.~\eqref{eq:sequence extension}.
Then, we randomly choose one as the query sequence $Q^+$,
and calculate DTW distance between $Q^+$ and every candidate sequence $C^+_i$.
Correspondingly, DTW distance between $Q$ and every original sequence $C_i$ is also calculated.
In Fig.~\ref{fig:GunPoint extension},
we can see $D_{dtw}(Q^+,C^+_i)$ is always less than or equal to $D_{dtw}(Q,C_i)$,
which just verifies Eq.~\eqref{eq:path extension theorem 1} in Theorem.~\ref{theorem:5}.
In fact, there is almost no difference between $D_{dtw}(Q,C_i)$ and $D_{dtw}(Q^+,C^+_i)$,
except for the two candidate sequences $C^+_{24}$, $C^+_{112}$.

\begin{figure}[htbp]
  \centering
  \includegraphics[width=0.49\textwidth]{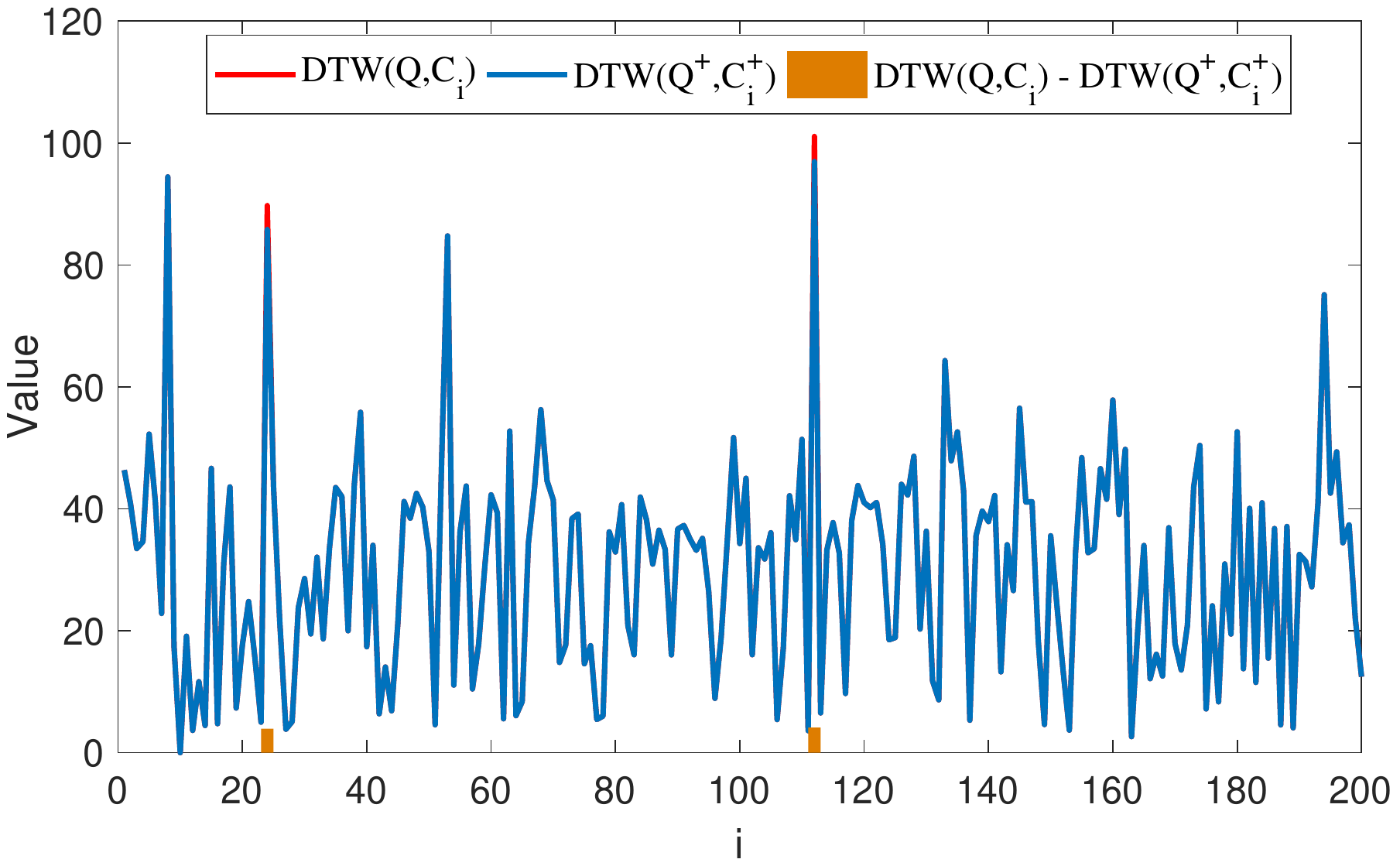}
  \caption{Results of sequence extension on GunPoint dataset.}
  \label{fig:GunPoint extension}
\end{figure}

We continue to increase $Lmax$ to observe the effect of extension length on DTW distance.
Fig.~\ref{fig:GunPoint further extension} shows the average values of $D_{dtw}(Q,C_i)$, $D_{dtw}(Q^+,C^+_i)$
between the query sequence and every candidate sequence.
With the increase of $Lmax$,
the average value of $D_{dtw}(Q^+,C^+_i)$ remains the same.
It indicates that extension length does not affect DTW distance,
which just verifies Eq.~\eqref{eq:path extension further theorem 1} in Theorem.~\ref{theorem:6}.

\begin{figure}[htbp]
  \centering
  \includegraphics[width=0.46\textwidth]{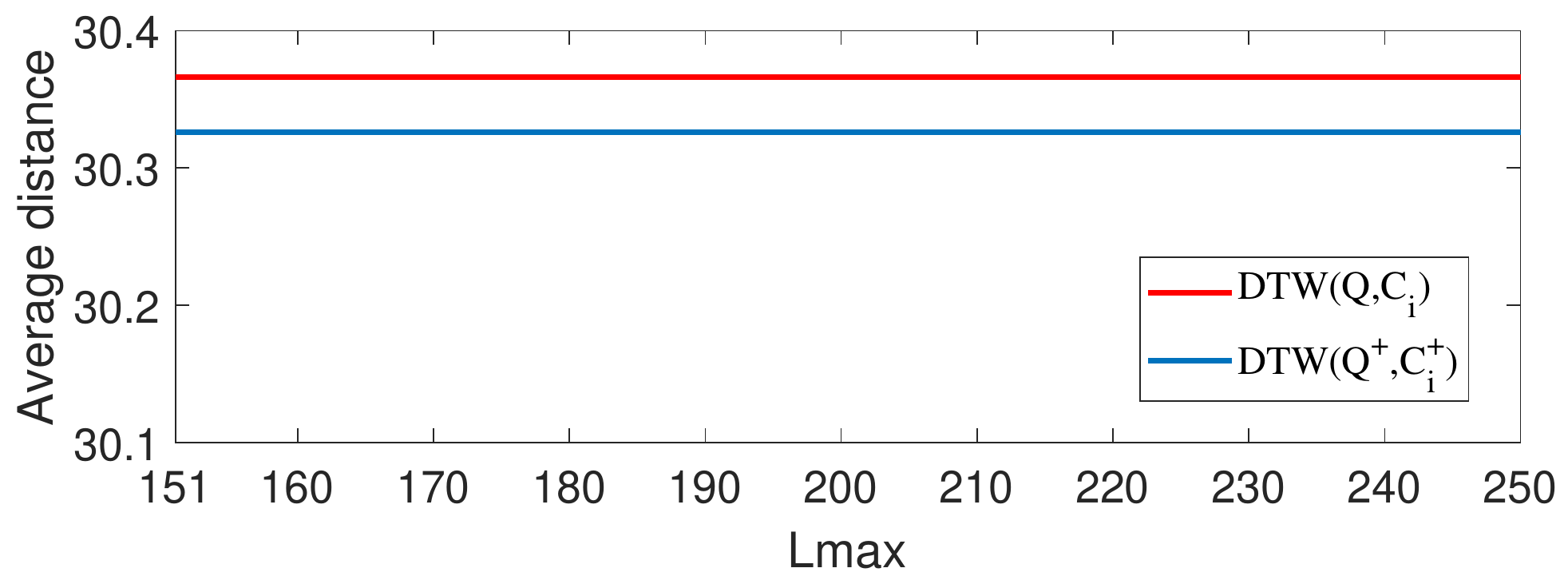}
  \caption{Results of increasing sequence length on GunPoint dataset.}
  \label{fig:GunPoint further extension}
\end{figure}

We repeat the experiments above on other datasets.
Fig.~\ref{fig:benchmark datasets further extension} illustrates the experimental results,
where the sequences in a dataset are all extended to the same minimum length.
On all the benchmark datasets,
the average values of $D_{dtw}(Q,C_i)$, $D_{dtw}(Q^+,C^+_i)$ are close to each other.
This means that sequence extension has little effect on DTW distance.

\begin{figure}[htbp]
  \centering
  \includegraphics[width=0.48\textwidth]{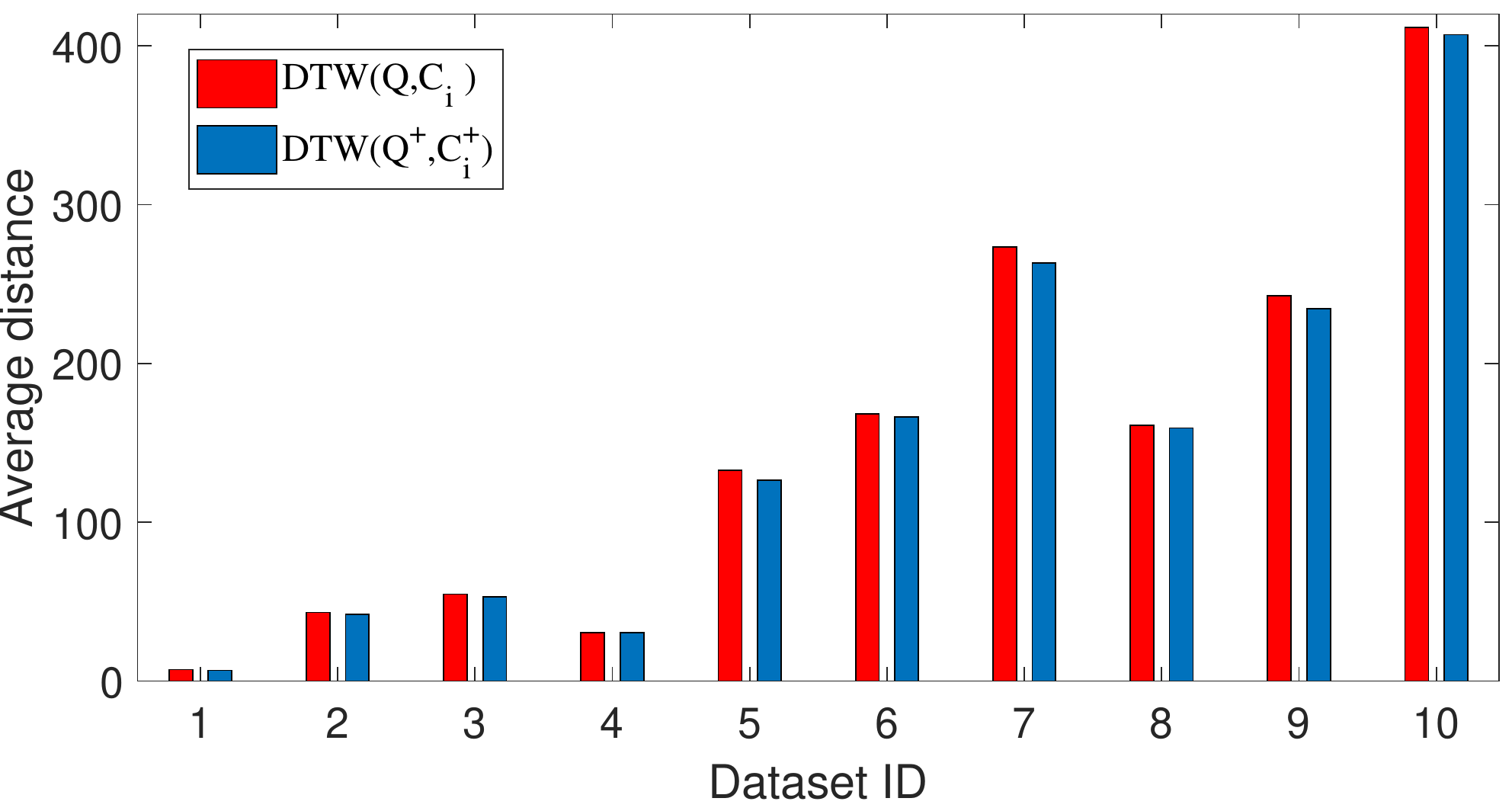}
  \caption{Results of sequence extension on 10 benchmark datasets.}
  \label{fig:benchmark datasets further extension}
\end{figure}

\subsection{Comparison of lower bounding distances}

Our proposed lower bounding distance is a seamless combination of sequence extension and LB\_Keogh.
Thus, we mark it as LB\_Keogh$^+$ in the paper.
We compare LB\_Keogh$^+$ with LB\_Kim and LB\_Yi,
which can be used for unequal-length sequences.

We first evaluate LB\_Keogh$^+$ with tightness $T\in$ [0,\,1],
which is defined as the ratio of the lower bounding distance to DTW distance.
The larger the tightness is, the better the lower bounding distance is.
The tightness of LB\_Keogh$^+$ is written as:
\begin{equation}\label{eq:tightness_proposed method}
\begin{split}
	T_{Keogh^+} = \frac{LB\_Keogh^+(Q,C_i)}{D_{dtw}(Q,C_i)}
\end{split}
\end{equation}

For an $\epsilon$-range search on a dataset,
we compute the average tightness between the query sequence $Q$ and each candidate sequence $C_i$.
Fig.~\ref{fig: tightness comparison on benchmark datasets} and Table~\ref{tab:tightness comparison on benchmark datasets}
illustrate the average tightness over 100 $\epsilon$-range searches on every dataset.

\begin{figure}[htbp]
  \centering
  \includegraphics[width=0.42\textwidth]{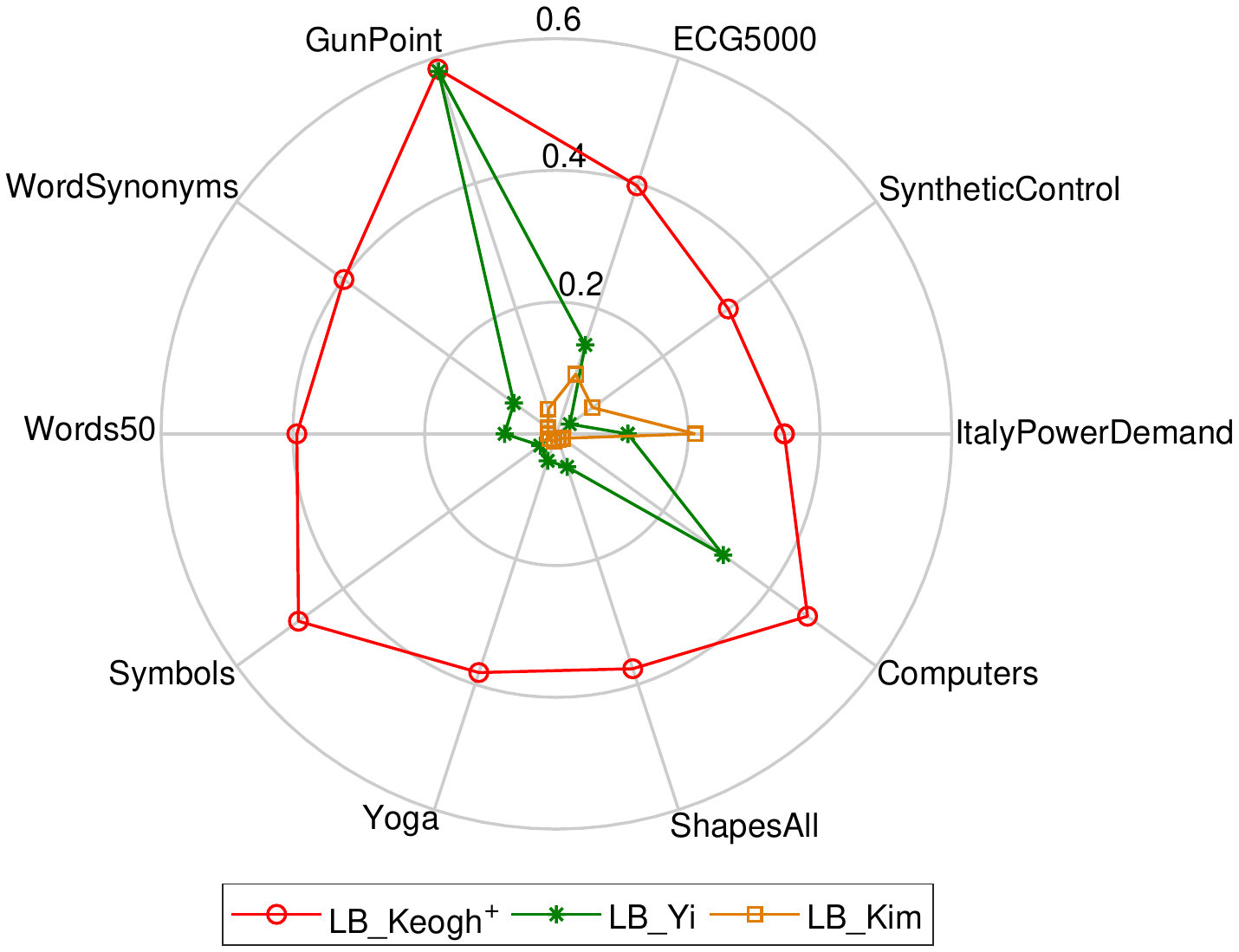}
  \caption{Tightness comparison on 10 benchmark datasets.}
  \label{fig: tightness comparison on benchmark datasets}
\end{figure}

We can see, the tightness of LB\_Keogh$^+$ is obviously greater than that of LB\_Kim and LB\_Yi on every dataset,
and LB\_Kim has the lowest tightness on most datasets except for ItalyPowerDemand.
The reason is that LB\_Keogh$^+$ makes more points of sequences participate in the calculation of lower bounding distance.
LB\_Kim only chooses four feature points to calculate the lower bounding distance.
When sequence length is short, such as ItalyPowerDemand, the tightness of LB\_Kim is acceptable.
With the increase of sequence length, its tightness decreases sharply.

\begin{table}[H]
\centering
\caption{Results of tightness comparison.}\label{tab:tightness comparison on benchmark datasets}
\begin{tabular}{p{2.3cm}p{1.5cm}<{\centering}p{1.2cm}<{\centering}p{1.2cm}<{\centering}}
\hline
\textbf{Dataset} & $\textbf{LB\_Keogh}^+$ &	$\textbf{LB\_Yi}$ &	$\textbf{LB\_Kim}$ \\
\hline
    ItalyPowerDemand & 0.3459 & 0.1083 & 0.2105 \\
    SyntheticControl & 0.3225 & 0.0251 & 0.0678 \\
    ECG5000 & 0.3959 & 0.1419 & 0.0952 \\
    GunPoint & 0.5823 & 0.5787 & 0.0396 \\
    WordSynonyms   & 0.3986 & 0.0798 & 0.0162 \\
    Words50 & 0.3939 & 0.0786 & 0.0122 \\
    Symbols & 0.4839 & 0.0302 & 0.0128 \\
    Yoga  & 0.3810 & 0.0431 & 0.0124 \\
    ShapesAll & 0.3749 & 0.0527 & 0.0092 \\
    Computers & 0.4711 & 0.3131 & 0.0115 \\
\hline
\end{tabular}
\end{table}

Pruning power $P\in$ [0,1] is another important indicator to evaluate lower bounding distances,
which is defined as:
\begin{equation}\label{eq:pruning power}
\begin{split}
	P = \frac{S_0}{S}
\end{split}
\end{equation}
where $S$ is the number of sequences calculated by DTW distance using sequential scan method,
$S_0$ is the number of sequences that do not require the calculation of DTW distance by using a kind of lower bounding distance.
The larger $P$ is, the better the filtering effect of a lower bounding distance is.

We calculate the average pruning power over 100 $\epsilon$-range searches on GunPoint, ECG5000, Yoga and Computers datasets.

ECG5000 dataset is a 20-hour long ECG downloaded from Physionet.
The data are pre-processed in two steps: extract each heartbeat; make each heartbeat equal length using interpolation.
After that, 5,000 heartbeats are randomly selected.
The patients have severe congestive heart failure and the class values are obtained by automated annotation.

Yoga dataset is obtained by capturing two actors transiting between yoga poses in front of a green screen,
as illustrated in Fig.~\ref{fig:Yoga}.
Each image was converted to a one dimensional series by finding the outline and measuring the distance of the outline to the centre.
The problem is to discriminate between one actor (male) and another (female).
\begin{figure}[htbp]
  \centering
  \includegraphics[width=0.49\textwidth]{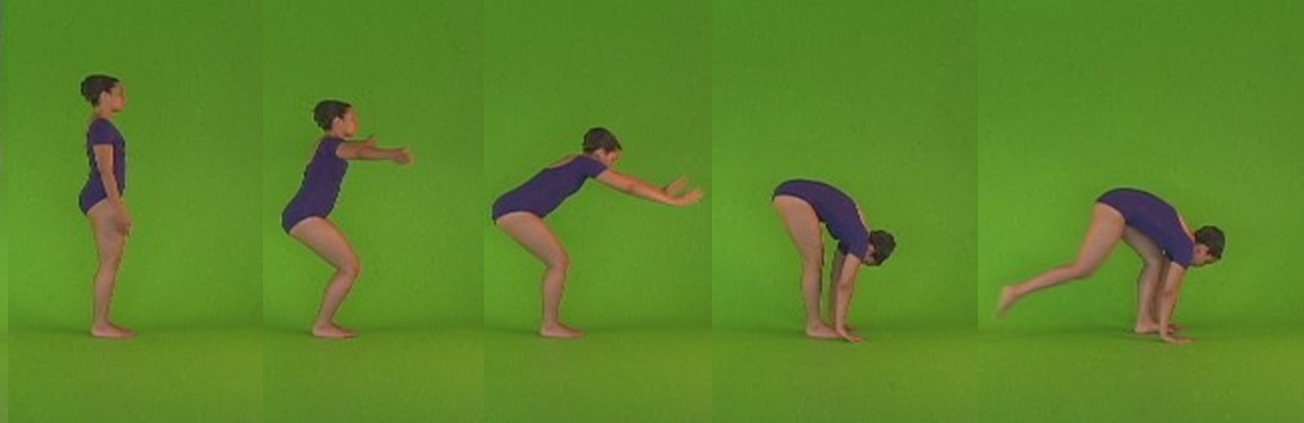}
  \caption{Description of Yoga dataset.}
  \label{fig:Yoga}
\end{figure}

Computers dataset is taken from data recorded as part of government sponsored study called Powering the Nation.
The intention is to collect behavioural data about how consumers use electricity within the home to help reduce the UK's carbon footprint.
The data contains readings from 251 households, sampled in two-minute intervals over a month.
Classes are Desktop and Laptop.

Experimental results are shown in Fig.~\ref{fig:comparison of pruning power},
where the height of a bar represents the ratio of the number of retrieved sequences to the total number of sequences.
We can see the pruning power of LB\_Keogh$^+$ outperforms that of LB\_Kim and LB\_Yi,
and LB\_Kim hardly works except when $\epsilon$ is small.
Because LB\_Kim uses $L_\infty$ instead of $L_1$ or $L_2$ as its base distance.
Unlike most of other methods, the DTW distance defined by LB\_Kim is not the sum of the base distance,
which makes it unfair to compare LB\_Kim with other methods.

\begin{figure*}[!htbp]
\centering
\subfloat[GunPoint]{\includegraphics[width= 7.2 cm, height = 4 cm ]{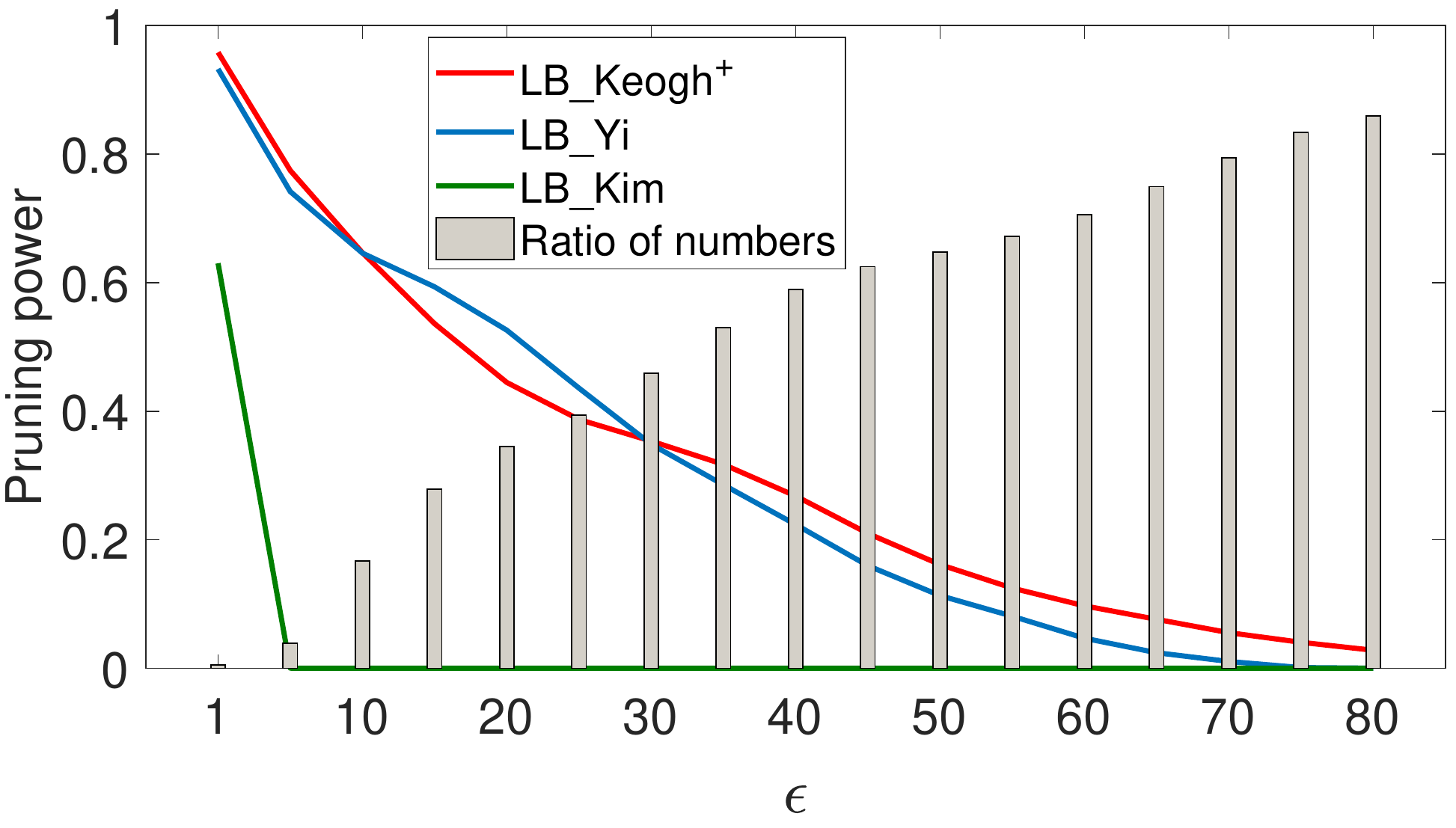}}
~~~~~~\subfloat[ECG5000]{\includegraphics[width= 7.2 cm, height = 4 cm ]{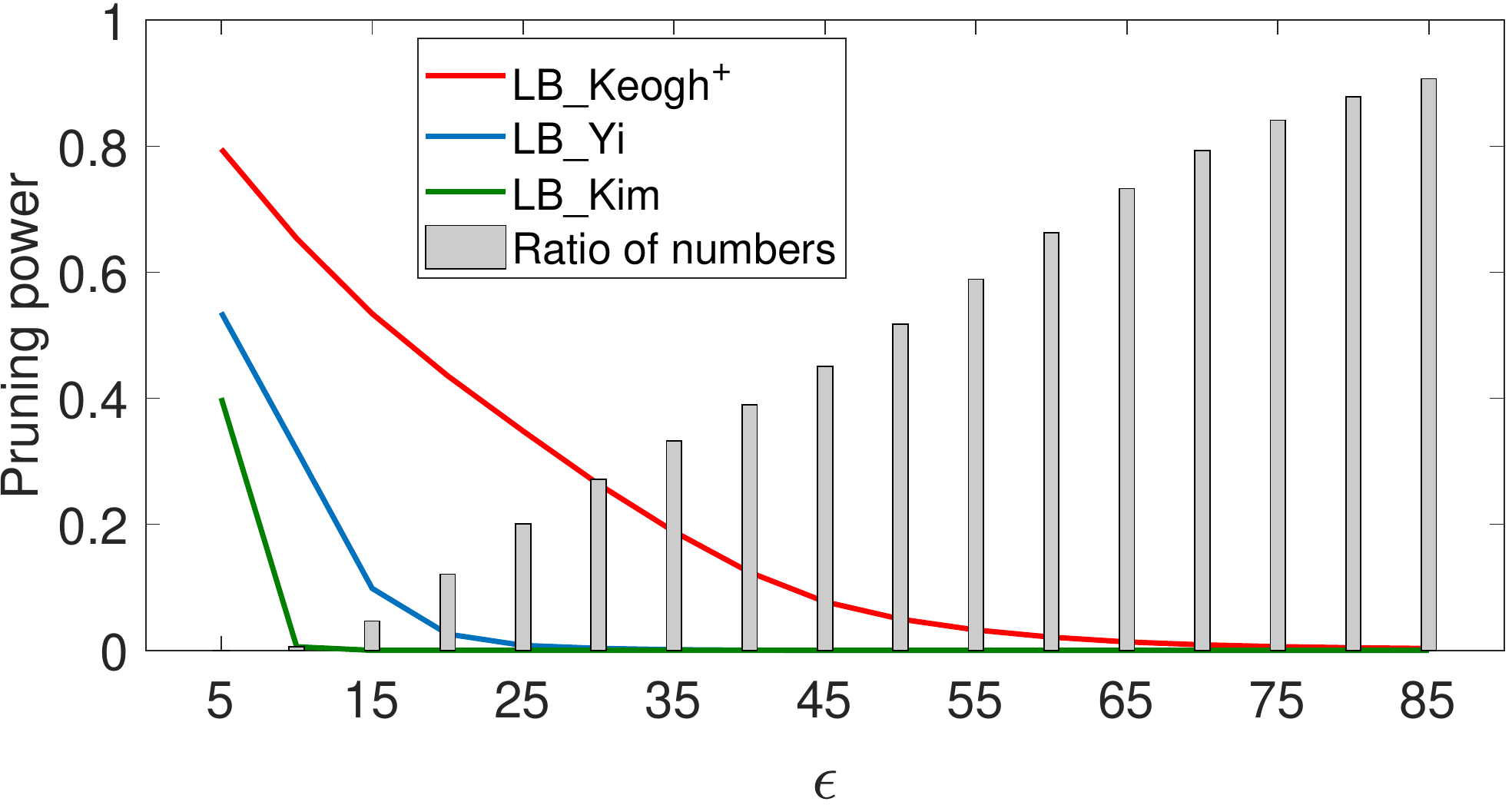}}\\
\subfloat[Yoga]{\includegraphics[width= 7.2 cm, height = 4 cm ]{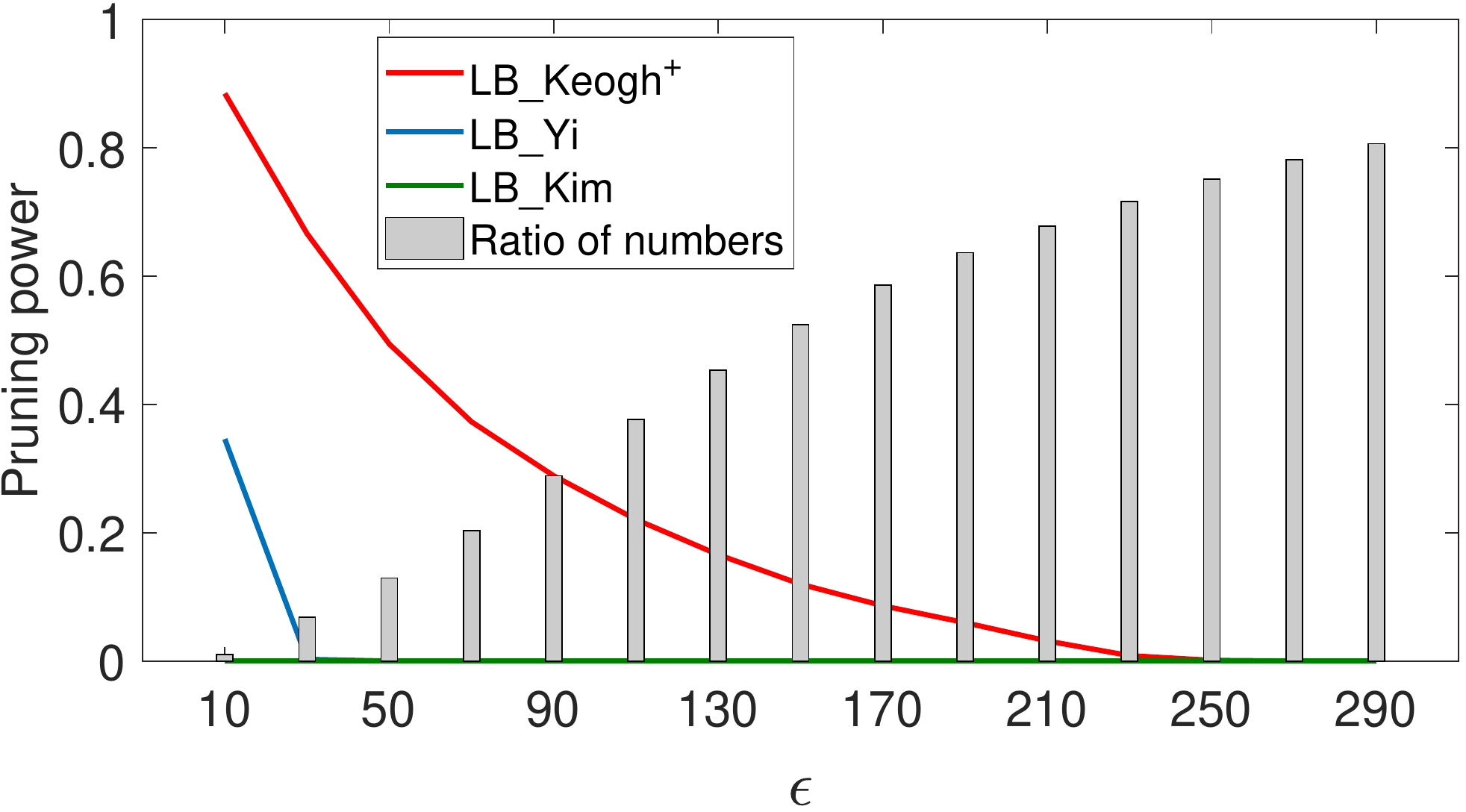}}
~~~~~~\subfloat[Computers]{\includegraphics[width= 7.2 cm, height = 4 cm ]{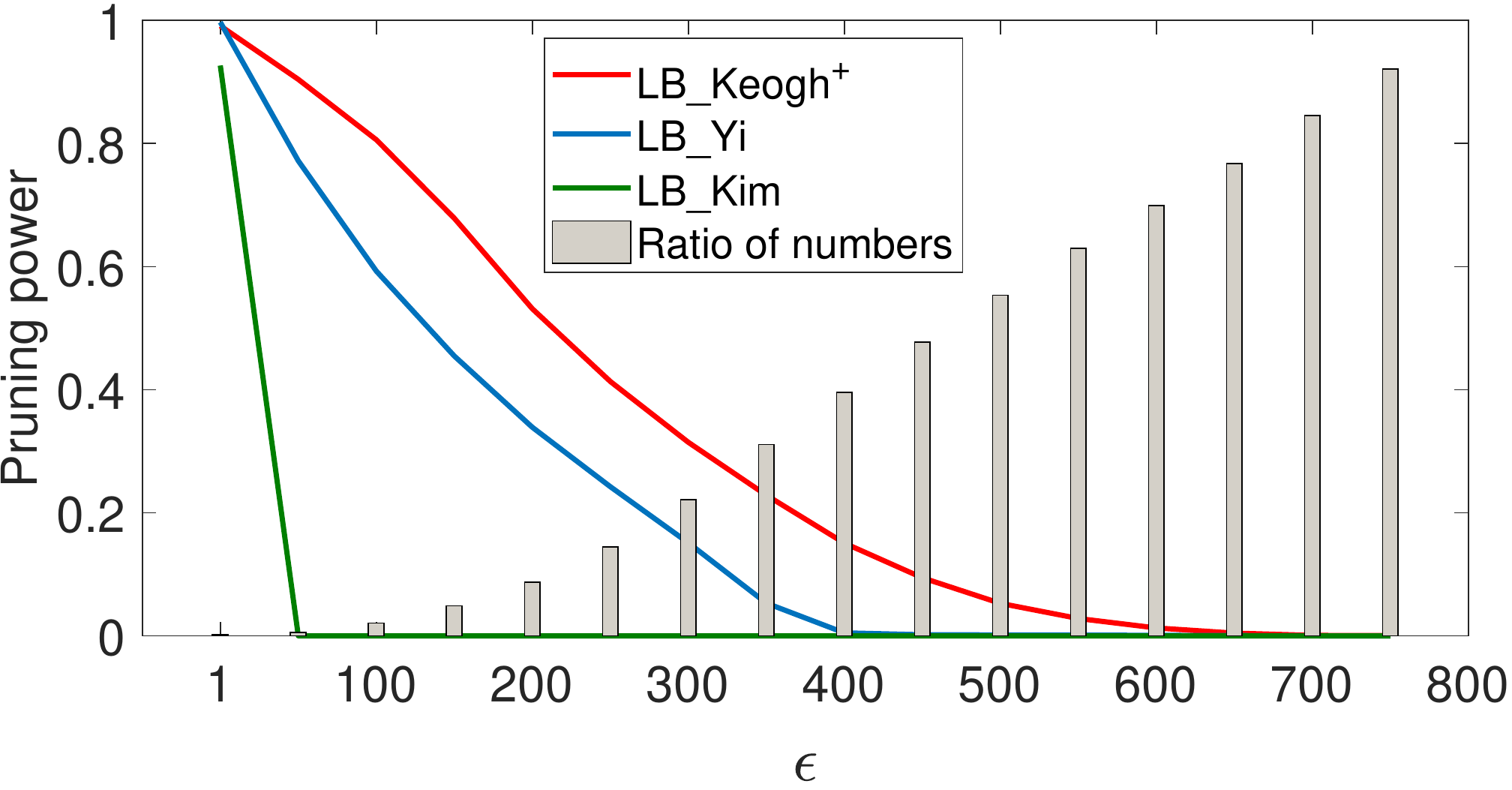}}
\caption{Comparison of pruning power on four benchmark datasets.}
\label{fig:comparison of pruning power}
\end{figure*}

When $\epsilon$ is small, the pruning power of LB\_Keogh$^+$ is acceptable.
With the increase of $\epsilon$, its pruning power also gradually decreases to zero.
The reason is that $\epsilon$ determines the number of sequences filtered out by LB\_Keogh$^+$.
Especially, when $\epsilon$ is large enough, LB\_Keogh$^+$ completely loses its filtering effect.
Generally speaking, $\epsilon$-range search only needs to find a small number of sequences from a dataset,
for example, 10\% of the total number of sequences.
That is, $\epsilon$ is usually not large in the sense.
Therefore, LB\_Keogh$^+$ has desirable pruning power in practical applications.

\subsection{Discussion on the number of subsegments in PAA}

In this subsection, we discuss the effect of changing the number of subsegments in PAA.
The tightness of LB\_{PAA} is defined as:
\begin{equation}\label{eq:tightness_paa}
\begin{split}
	T_{PAA} = \frac{LB\_{PAA}(\bar{U^+},\bar{L^+},\bar{C_i^+})}{D_{dtw}(Q,C_i)}
\end{split}
\end{equation}

We extend all the sequences of a dataset to the same length $Lmax$,
such that more integers are divisible by this length.
The results of sequence extension are shown in Table~\ref{tab:extended length}.
Then, we calculate the average tightness over 100 $\epsilon$-range searches,
as illustrated in Fig.~\ref{fig:subsegments in PAA}.

\begin{table}[H]
\centering
\caption{The extended length of the sequences.}\label{tab:extended length}
\begin{tabular}{p{2.1cm}p{1cm}<{\centering}|p{1.2cm}p{1cm}<{\centering}}
\hline
\textbf{Dataset} & \textbf{Lmax} & \textbf{Dataset}  &	\textbf{Lmax} \\
\hline
ItalyPowerDemand   &26 &Words50   &272 \\
SyntheticControl   &63 &Symbols   &399 \\
ECG5000   &144 &Yoga   &429 \\
GunPoint   &152 &ShapesAll   &513 \\
WordSynonyms   &272 &Computers  &726 \\
\hline
\end{tabular}
\end{table}

\begin{figure}[!htbp]
\centering
\subfloat[Average length $\leq$ 256]{\includegraphics[width= 4.5 cm, height = 3.8 cm ]{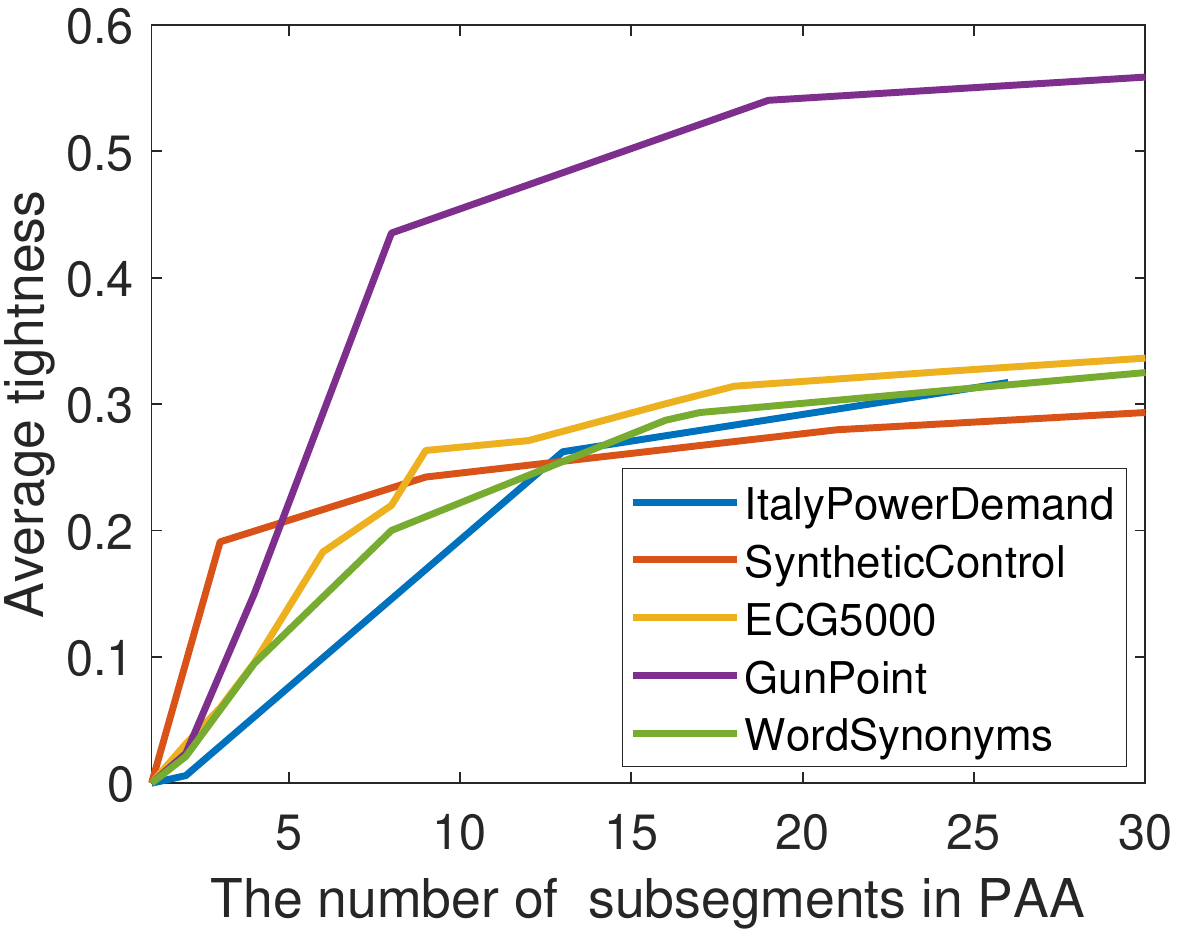}}
\subfloat[Average length $\geq$ 257]{\includegraphics[width= 4.5 cm, height = 3.8 cm ]{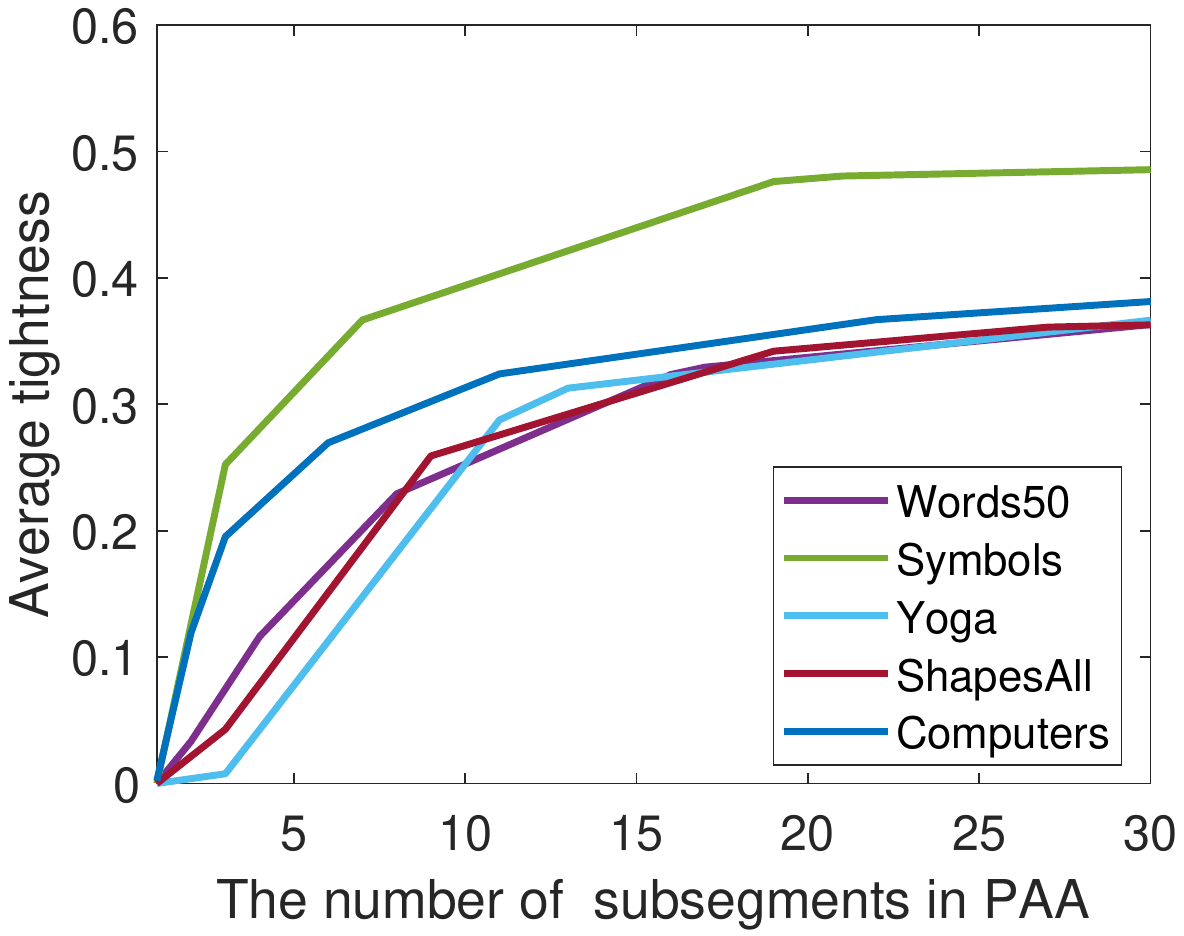}}
\caption{Experimental results on the number of subsegments in PAA.}
\label{fig:subsegments in PAA}
\end{figure}

With the increase of the number of subsegments,
the tightness of LB\_{PAA} goes up gradually.
Because the more subsegments there are,
the more approximate $\bar{U^+}$, $\bar{L^+}$, $\bar{C^+}$ are to $U^+$, $L^+$, $C^+$.
When the number reaches a turning point, such as 16,
the growth trend of tightness become flat.
That is, if the number of subsegments is chosen at the turning point,
we can dramatically reduce the dimension of time series
and make the tightness of LB\_{PAA} acceptable.

We further analyze the effect of the length of sequence extension on LB\_{PAA}.
According to Eq.~\eqref{eq:tightness_paa},
we calculate the average tightness of LB\_{PAA} over 100 $\epsilon$-range searches.
In Fig.~\ref{fig:Effect of increasing length on LBPAA},
we can see from the micro point of view,
when the number of subsegments is fixed,
the longer the sequences are extended, the smaller the tightness will be.
The reason is that if the number of subsegments is fixed,
the approximation of PAA to sequences becomes worse with the increase of the length of sequences.
From the macro point of view,
the variation of extension length within a certain range has little effect on the tightness of LB\_{PAA}.
Therefore, we can extend sequences to an arbitrary suitable length,
such that the length is just an integer multiple of the desired number of subsegments.

\begin{figure}[!htbp]
\centering
\subfloat[ECG5000]{\includegraphics[width= 4.5 cm, height = 3.4 cm ]{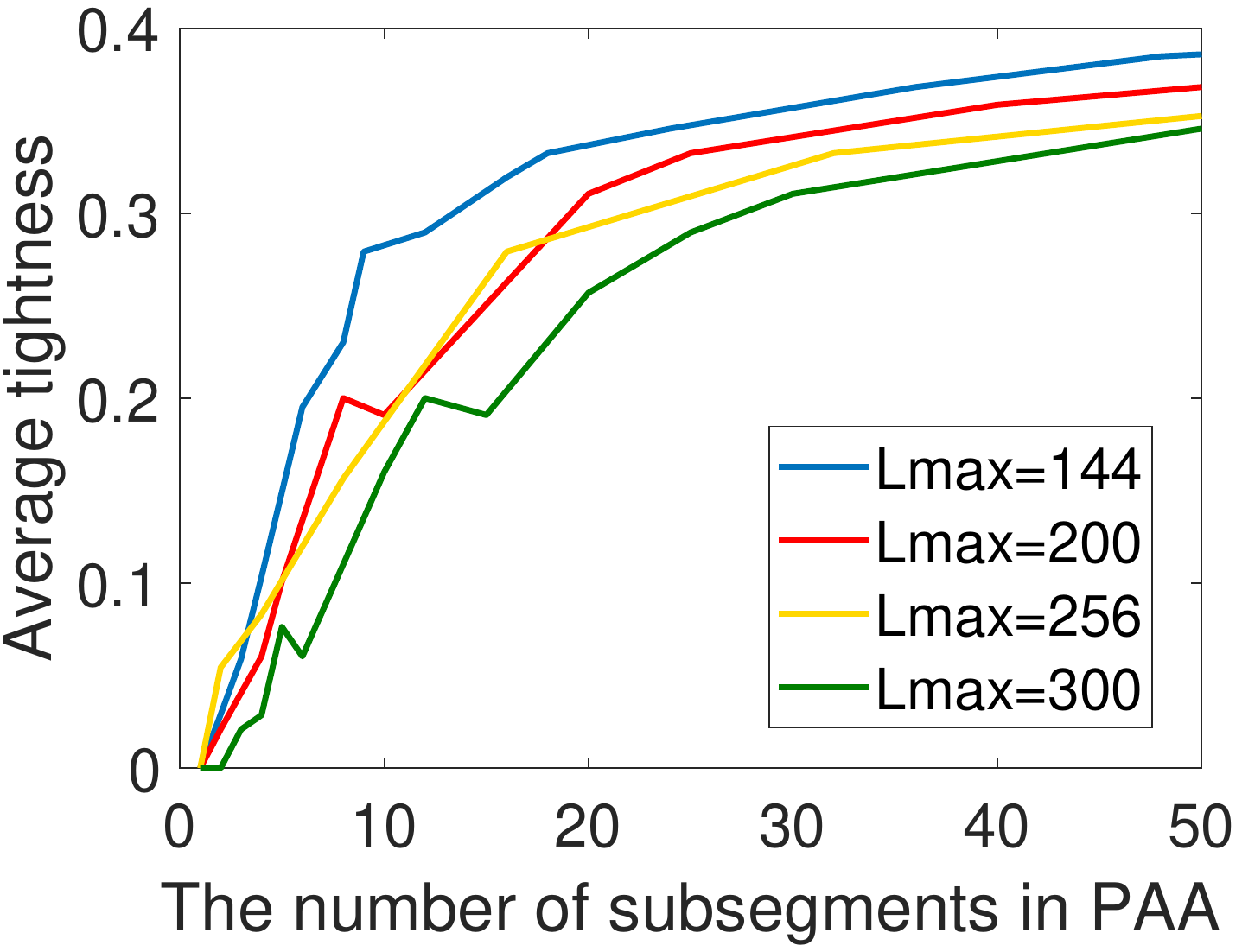}}
~\subfloat[GunPoint]{\includegraphics[width= 4.5 cm, height = 3.4 cm ]{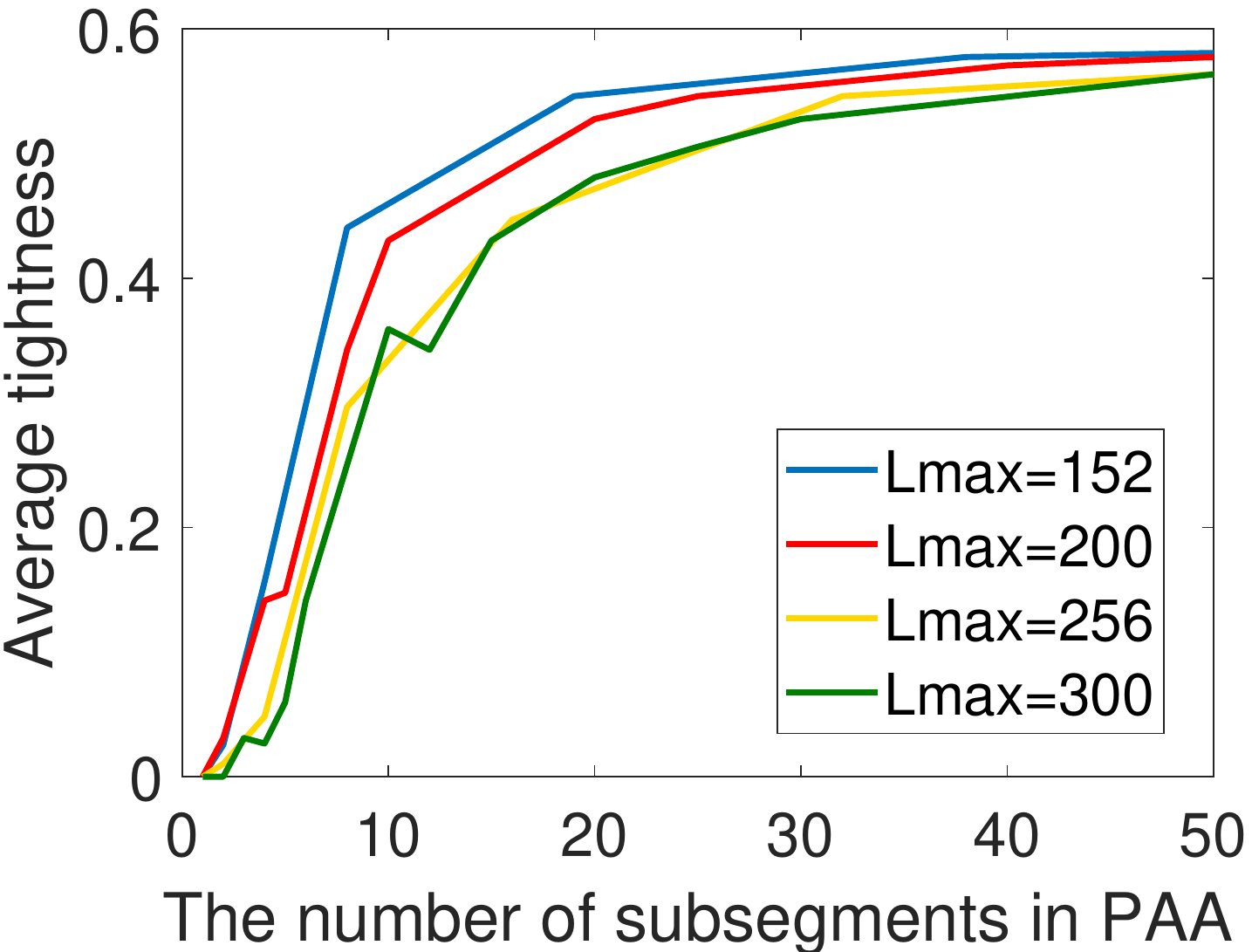}}\\
\subfloat[Yoga]{\includegraphics[width= 4.5 cm, height = 3.4 cm ]{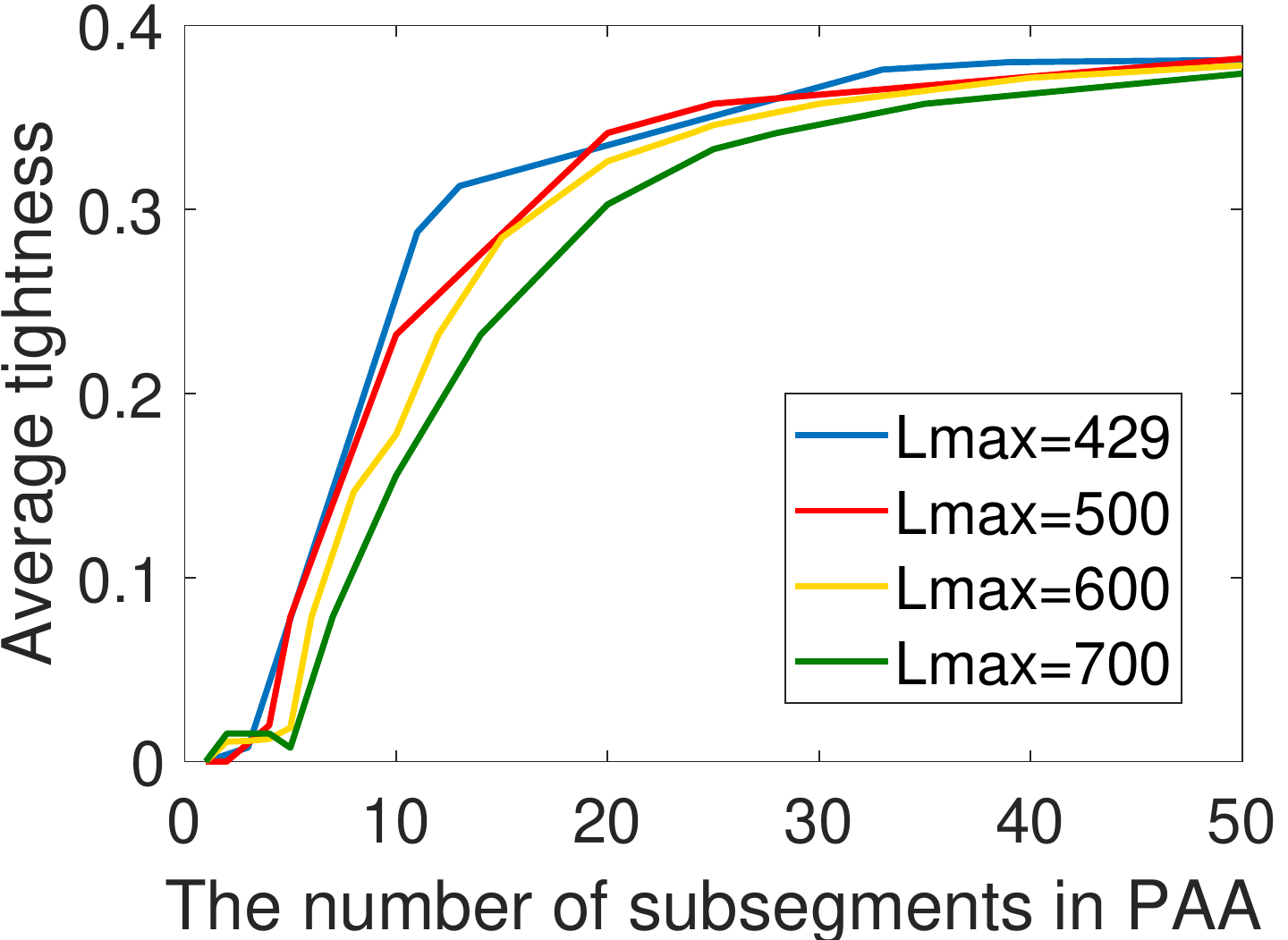}}
~\subfloat[Computers]{\includegraphics[width= 4.5 cm, height = 3.4 cm ]{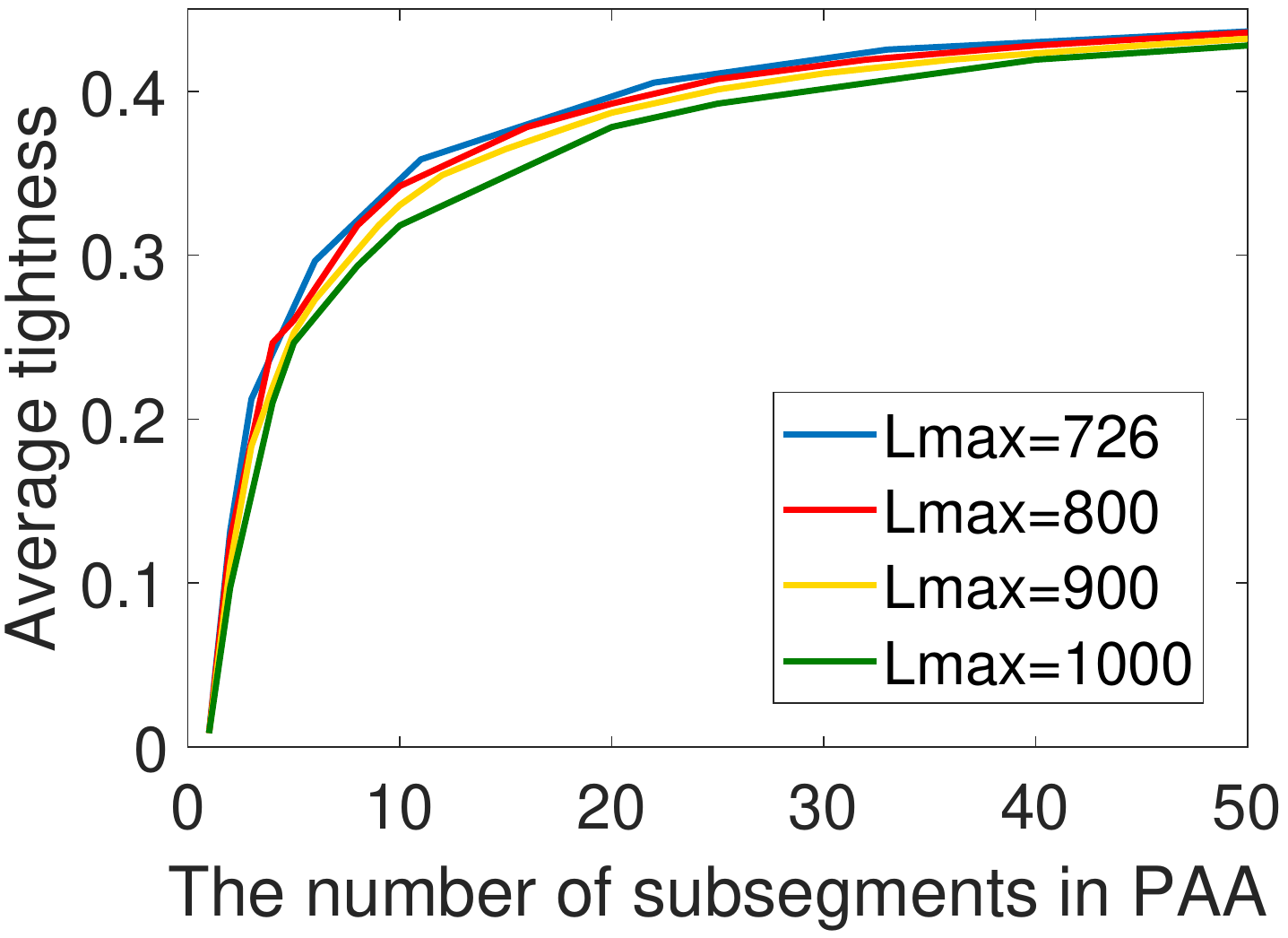}}
\caption{Effect of increasing sequence length on LB\_PAA.}
\label{fig:Effect of increasing length on LBPAA}
\end{figure}

\subsection{The effect of changing the warping window width}

In the subsection, we discuss the effect of changing the warping window width.
In order to represent time series with 16-dimensional vectors,
we can easily extend all the sequences of a dataset to the same length,
such that the length is the smallest integer multiple of 16.

Then, we calculate the average tightness of LB\_Keogh$^+$ and LB\_{PAA} under different constraints $r$ of Sakoe-Chiba band,
by repeating the experiments above.
We can see from Fig.~\ref{fig:change the warping window width},
with the increase of $r$, the tightness of LB\_Keogh$^+$ and LB\_{PAA} decreases gradually.
Because if $r$ changes from 10\%  to 20\%,
the area of the shadow parts in Fig.~\ref{fig:the lower bounding distance} will reduce.
Accordingly, LB\_Keogh$^+$ and LB\_{PAA} are bound to decrease.

\begin{figure}[!htbp]
\centering
\subfloat[Effect on LB\_{Keogh}$^{+}$]{\includegraphics[width= 4.3 cm, height = 3.5 cm ]{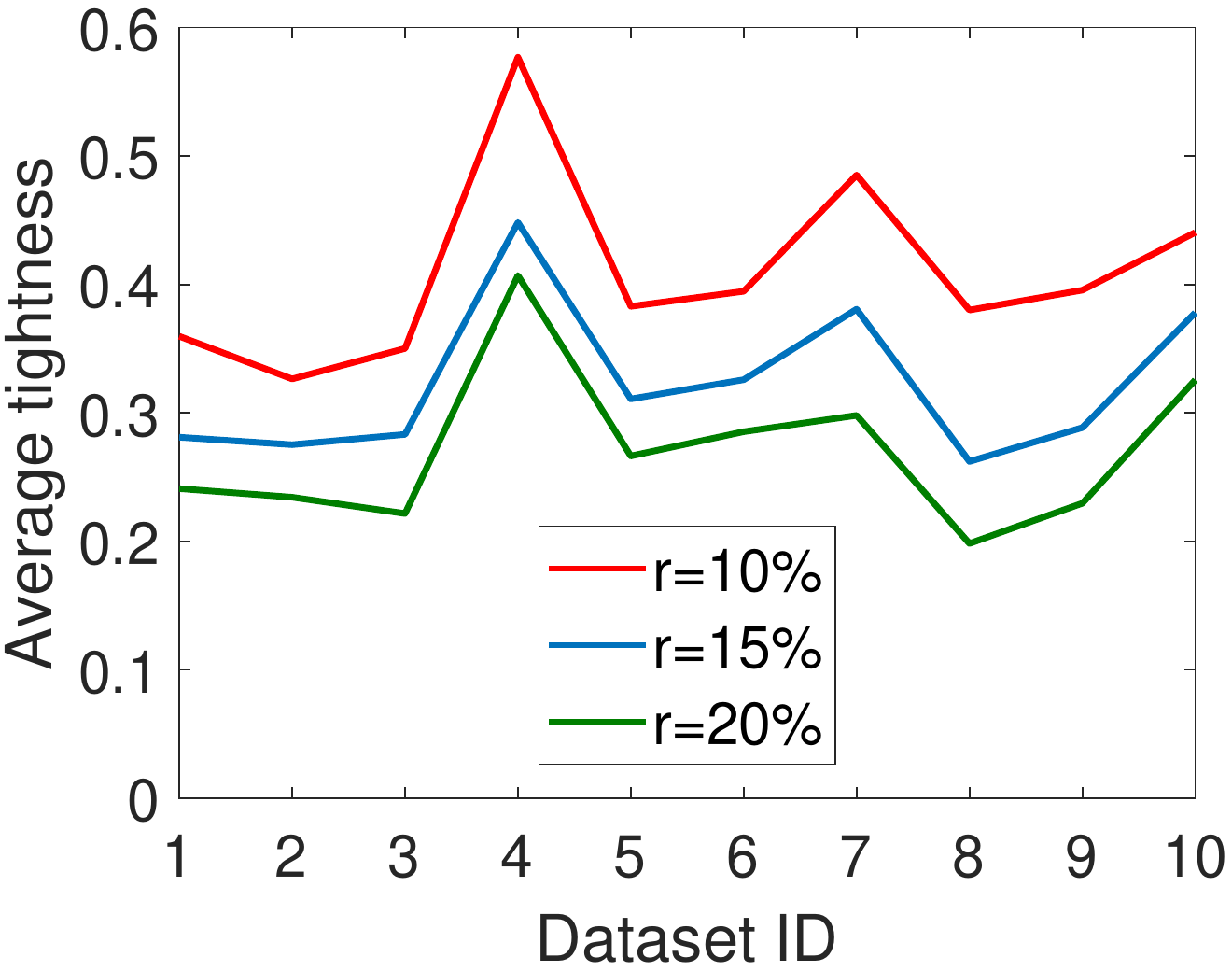}}
~\subfloat[Effect on LB\_PAA]{\includegraphics[width= 4.3 cm, height = 3.5 cm ]{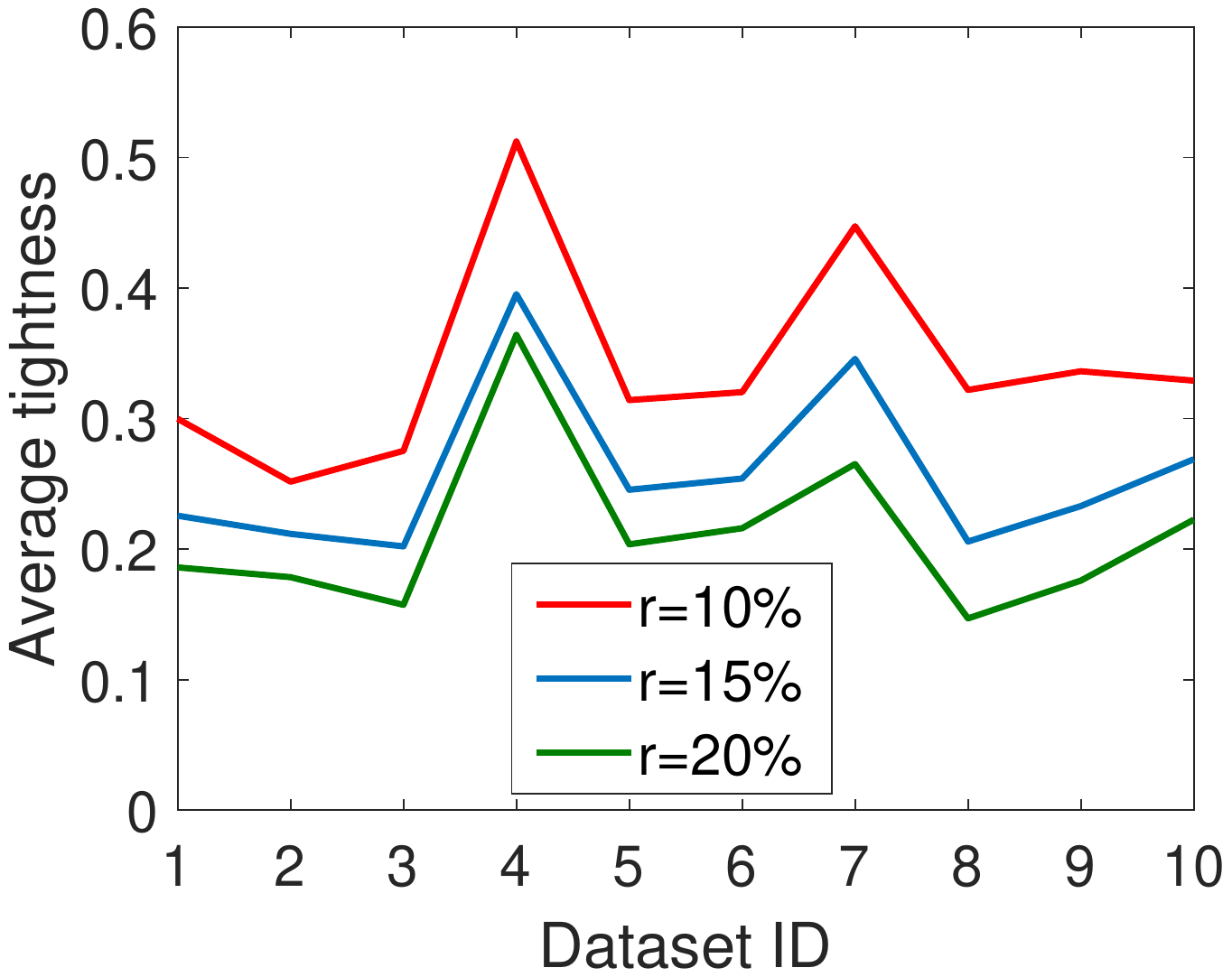}}\\
\caption{Experimental results of changing the warping window width.}
\label{fig:change the warping window width}
\end{figure}

\section{Conclusions and future work}\label{sec:conclusions}

In the paper, we propose a novel lower bounding distance LB\_Keogh$^+$,
which is a seamless combination of sequence extension and LB\_Keogh.
It can be used for unequal-length sequences and has low computational complexity.
Besides, LB\_Keogh$^+$ can extend sequences to an arbitrary suitable length,
such that the length is exactly an integer multiple of the desired number of subsegments in PAA.
Next, based on LB\_Keogh$^+$, an exact indexing of time series under DTW is devised.
Then, we introduce several theorems and complete the relevant proofs
to guarantee no false dismissals in similarity search.
Finally, extensive experiments are conducted to evaluate the proposed method on real-world datasets.
The experimental results indicate that
the proposed method can effectively perform similarity search of unequal-length time series
with high tightness and good pruning power.

Besides, there are still some interesting contents that deserve further study.
For instance, how to devise a kind of effective index for multivariate time series
is a promising research direction.


%

%

\ifCLASSOPTIONcompsoc
  \section*{Acknowledgments}
\else
  \section*{Acknowledgment}
\fi

Thanks to the donors who have made contributions to the benchmark datasets.
We would like to sincerely thank Eamonn Keogh, Chotirat Ann Ratanamahatana, \emph{etc}. for their inspiring work in this field.
This work is supported by the National Natural Science Foundation of China under Grant No. 61502521.

\ifCLASSOPTIONcaptionsoff
  \newpage
\fi

\bibliographystyle{IEEEtran}
\bibliography{SWC}



%
%
%

\begin{IEEEbiography}[{\includegraphics[width=1in,height=1.25in,clip,keepaspectratio]{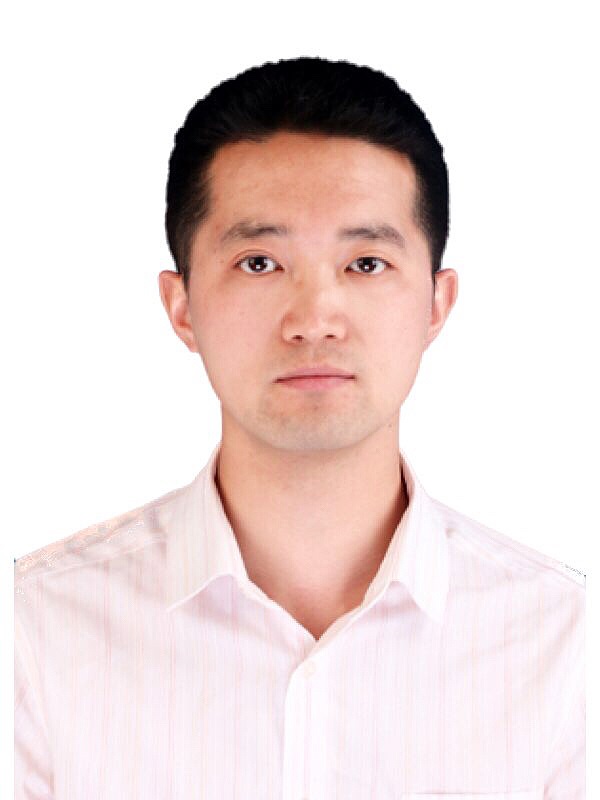}}]{Zhengxin Li}
received the Ph.D. degree in Control Science and Engineering from Air Force Engineering University, China in 2011.
He is currently a postdoctor in the School of Computer Science
and Center for OPTical IMagery Analysis and Learning (OPTIMAL),
Northwestern Polytechnical University.
His research interests mainly include time series pattern recognition, machine learning and data mining.
\end{IEEEbiography}

%






\end{document}